\documentclass[10pt]{article}

\usepackage{lmodern}
\usepackage{textcomp}
\usepackage[margin=1in]{geometry}
\usepackage{amsmath,amssymb,amsthm,mathtools}
\usepackage{booktabs}
\usepackage[numbers,sort&compress]{natbib}
\usepackage{graphicx}
\usepackage{subcaption}
\usepackage{xcolor}
\usepackage{tikz}
\usepackage{hyperref}
\usetikzlibrary{arrows.meta,positioning}
\graphicspath{{figures/}}
\newcommand{\figdir}{figures}

\newcommand{\eg}{e.g.,\ }
\newcommand{\ie}{i.e.,\ }

\theoremstyle{plain}
\newtheorem{proposition}{Proposition}
\newtheorem{lemma}{Lemma}
\newtheorem{definition}{Definition}

\theoremstyle{remark}
\newtheorem{remark}{Remark}

\title{Autodeleveraging as Online Learning}
\author{Tarun Chitra \and Nagu Thogiti \and Mauricio Jean Pieer Trujillo Ramirez \and Victor Xu}
\date{February 2026}

\begin{document}
\maketitle

\begin{abstract}
Autodeleveraging (ADL) is a last-resort loss socialization mechanism used by perpetual futures venues when liquidation and insurance buffers are insufficient to restore solvency.
Despite the scale of perpetual futures markets, ADL has received limited formal treatment as a sequential control problem.
This paper provides a concise formalization of ADL as online learning on a PNL-haircut domain: at each round, the venue selects a solvency budget and a set of profitable trader accounts.
The profitable accounts are liquidated to cover shortfalls up to the solvency budget, with the aim of recovering exchange-wide solvency.
In this model, ADL haircuts apply to positive PNL (unrealized gains), not to posted collateral principal.
Using our online learning model, we provide robustness results and theoretical upper bounds on how poorly a mechanism can perform at recovering solvency.
We apply our model to the October 10, 2025 Hyperliquid stress episode.
The regret caused by Hyperliquid's production ADL queue is about 50\% of an upper bound on regret, calibrated to this event, while our optimized algorithm achieves about 2.6\% of the same bound.
In dollar terms, the production ADL model over liquidates trader profits by up to \$51.7M.
We also counterfactually evaluated algorithms inspired by our online learning framework that perform better and found that the best algorithm reduces overshoot to \$3M.
Our results provide simple, implementable mechanisms for improving ADL in live perpetuals exchanges.
\end{abstract}

\section{Introduction}
Perpetual futures (or simply, perpetuals) are expiryless derivatives that provide linear, delta-one exposure to an underlying asset with margin instead of full notional outlay.
Economically, they are close to contracts for difference (CFDs): participants exchange value as marks move, rather than exchanging the underlying asset itself.
This design is why perpetuals are both simple to reason about and capital-efficient for directional risk transfer, and it explains their dominance in crypto derivatives trading~\citep{AngerisChitra2023PerpsSIAM}.
In 2024, centralized exchanges processed roughly~\$58.5T in perpetuals notional versus roughly~\$17.4T in spot volume, a ratio of about 3.3x in favor of perpetuals~\citep{CoinGecko2025Perps,CoinGecko2025Q1}.

Perpetuals are popular because traders can scale exposure linearly with leverage while avoiding duration risk (e.g.~contract roll schedules).
But the absence of expiry does not remove risk; it changes how risk must be managed.
In particular, perpetuals replace expiry settlement with continuous margining and funding-rate transfers, so both traders and the venue must manage solvency continuously rather than only at settlement checkpoints.

\paragraph{Risk Management.}
Compared to dated futures, collateral management in perpetuals is more path-dependent.
Dated futures clear at expiry and have standard variation margin cycles, whereas perpetuals remain open indefinitely and use funding payments to tether perpetual future prices to spot.
A positive funding rate transfers value from longs to shorts, while a negative rate reverses the direction.
This creates a carry cost (or carry income) that forces users to actively manage leverage, holding horizon, and basis exposure over time~\citep{AngerisChitra2023PerpsSIAM}.

Risk management needs to be viewed from both the microeconomic (trader) perspective and the macroeconomic perspective of the exchange.
At the trader level, solvency means a position's assets exceed its liabilities, which is enforced by margin requirements, liquidations, and funding transfers.
At the venue level, solvency means aggregate trader claims remain supportable at the exchange level.
We formalize this in~\S\ref{sec:background}.

\paragraph{Residual Loss Management.}
It is possible for venue-level risk management to fail, leading to a net insolvency.
When an exchange is in such an insolvent state, traders cannot withdraw their full collateral and earned profits.
Define the exchange's shortfall as the positive gap between exchange liabilities and exchange assets.
Residual loss management in the face of insolvencies involves two options for resolving how much traders can withdraw from the exchange: how much of that shortfall does the exchange cover immediately, and which traders bears those losses.
This separation is useful because two venues can cover the same shortfall (in dollar terms) while inducing very different fairness and incentive outcomes depending on how the burden is allocated.

Traditional finance has analogous machinery in the form of central clearing parties (CCPs). 
CCP utilize techniques such as default waterfalls and variation-margin gains haircutting (VMGH) to allocate residual losses across surviving members under predefined rules~\citep{DuffieZhu2011,CPMI_IOSCO_2014}.
Some allocations are approximately proportional to exposure or gains, while others are layered by membership class, collateral seniority, or administrative uplift rules, which is explicitly non-pro-rata in burden~\citep{ERCOT2022CRR}.
Perpetuals venues face the same design problem under faster time constraints and thinner legal settlement buffers due to the bearer asset nature of cryptocurrencies. 

\paragraph{Autodeleveraging.}
Autodeleveraging (ADL) is an algorithmic form of residual loss management: the venue enforces a deterministic rule to reduce profits from winning positions, utilizing these seized profits to close a shortfall in real time.
Operationally, this is a positive-PNL haircut: a winner can lose part of unrealized gains while posted collateral principal remains protected.
In most traditional finance settings, residual losses are reconciled asynchronously through legal and institutional processes after the event.
For bearer, self-custodied crypto systems, delayed discretionary recovery is often infeasible and solvency restoration must be executable by rule at event time.

ADL directly implements both parts of residual-loss design: \emph{severity} (how much shortfall is covered now) and \emph{allocation} (which winners are deleveraged first and by how much).
While there are numerous allocation rules that can be used to implement ADL, two rules dominate: the rules used by the exchanges Binance and Hyperliquid.
Binance-style ADL uses a queue priority based on profitability and leverage to rank who is deleveraged first~\citep{BinanceADL}.
Hyperliquid publishes deterministic liquidation and ADL logic in venue documentation and executes these transitions on a publicly readable state machine~\citep{HyperliquidDocsLiquidations}.
When residual losses occur repeatedly, usage of ADL can naturally be formulated as an online decision problem: each round reveals realized shortfall and market state, a policy chooses an allocation, and performance is measured by cumulative regret against benchmark policies.
To our knowledge, formal analysis of this problem is still sparse, with the paper~\citep{ChitraEtAl2026ADLPreprint} providing the most complete treatment to date giving an ADL trilemma and a broader mechanism space than production queue rules.

\paragraph{October 10, 2025.}
The October 10--11, 2025 stress episode is the key empirical motivation for this paper.
Contemporaneous reporting describes one of the largest liquidation cascades in crypto derivatives and broad ADL activation across major venues, including Hyperliquid and Binance~\citep{CoinDesk2025LargestLiquidations,BinanceADL}.
The critical difference for measurement is observability: Binance is centralized, so researchers do not get a fully auditable public state for each ADL decision, while Hyperliquid exposes enough public state to reconstruct event-time mechanics in detail.

Using the public Hyperliquid reconstruction artifacts, the event spans 21:16--21:27 UTC on October 10, 2025, with \$2{,}103{,}111{,}431 in observed liquidations.
One can view roughly 16 iterative instances of ADL usage over that duration.
ADL, historically, was designed as a last resort mechanism that is only used a once to recover solvency as opposed to being used sequentially.
A natural series of questions arise: How do we analyze the impact of such repeated usage of ADL? Does the mechanism increase or reduce solvency upon iterated usage? What is the optimal policy for repeated ADL usage?

\paragraph{This paper.}
We formalize repeated ADL as an online decision problem over rounds and evaluate a two-term objective that can be measured from public replay data: tracking error and burden concentration.
We then analyze this objective theoretically and empirically with dynamic regret.
Relative to~\citep{ChitraEtAl2026ADLPreprint}, our contributions are:
\begin{itemize}
  \item an explicit public-data observation model and clear replay assumptions for policy comparison,
  \item regret-based evaluation focused on deployable policy classes (with ex post reference policies reported separately),
  \item queue-stability diagnostics linked to convex extreme-point instability (monotonicity violations and adjacent-round rank stability),
  \item an execution-price estimation extension that separates ex ante and ex post benchmarks and yields a bound on cumulative execution-induced failure \(V_T\).
\end{itemize}

Formally, we show that fixed-priority queues can incur linear regret, while adaptive controllers for the same objective can achieve sublinear regret.
We also analyze robustness under execution-price uncertainty.
Relative to~\citep{ChitraEtAl2026ADLPreprint}, our key theoretical addition is an explicit severity-variation term \(P_T^\theta\), which links regret to ex ante/ex post execution mismatch and closeout-liquidity instability.
This directly addresses critique points about benchmark interpretation, timing of available information, and execution sensitivity raised in~\citep{Thogiti2025ADL653M,Thogiti2025TrilemmaCritique,Thogiti2026ADLExplainer}.

On the empirical side, under an explicit public-data replay model for October 10, 2025, we find that Hyperliquid's production ADL policy overshoots by between \$45.0M and \$51.7M.
We compare observed performance to an instance-calibrated regret bound from Proposition~\ref{prop:severity-explicit}.
Using this calibrated scale, production reaches about \(50.0\%\) of the bound.
Our no-regret online algorithm is vector mirror descent (using only start-of-round information), which reaches about \(3.4\%\) of the same bound.
Our analysis improves the empirical analysis of~\citep{ChitraEtAl2026ADLPreprint}, where the authors did not explicitly compare empirical performance to theory.

\section{Background}\label{sec:background}

\subsection{Perpetuals exchanges and funding mechanics}
A \emph{perpetuals futures exchange} is a derivatives venue that allows traders to take leveraged long or short exposure to a continuous, expiryless futures contract.
Traders post collateral (margin) to keep positions solvent, and the venue uses a funding-rate transfer between longs and shorts to keep futures prices aligned with spot.

Each trader $i$ has an account and we denote the set of active accounts at time $t$ by $\mathcal{I}_t$.
For each account \(i\), let their signed position size be \(q_{i,t}\).
An account $i \in \mathcal{I}_t$ if $q_{i,t} \neq 0$.
Similarly, let the mark price be \(m_t\), and index/spot reference be \(s_t\) at time $t$.
The distinction between \(m_t\) and \(s_t\) is operationally important.
The index/spot reference \(s_t\) is an external anchor (typically a composite spot index), while the mark \(m_t\) is the venue's internal fair-price input for margin and liquidation logic.
Funding is designed to control the basis \(m_t-s_t\), but the two prices can diverge during stress.
Ignoring fees for notation, one-step price PNL is \(\Delta \mathrm{PNL}_{i,t}^{\mathrm{price}} = q_{i,t}(m_t-m_{t-1})\).
Perpetuals additionally include funding transfers to keep futures and spot aligned. A reduced-form update is \(\Delta \mathrm{PNL}_{i,t}^{\mathrm{fund}} = -f_t\, q_{i,t}\, m_t\),
where \(f_t\) is signed funding (positive values transfer from longs to shorts).
The total PNL increment is therefore
\begin{equation}
\Delta \mathrm{PNL}_{i,t}
=\Delta \mathrm{PNL}_{i,t}^{\mathrm{price}}+\Delta \mathrm{PNL}_{i,t}^{\mathrm{fund}},
\qquad
\mathrm{PNL}_{i,t}:=\sum_{\tau\le t}\Delta \mathrm{PNL}_{i,\tau}.
\end{equation}
We note that this view of PNL accounting follows~\citep{AngerisChitra2023PerpsSIAM,he2022fundamentals}.

\subsection{Assets, liabilities, and solvency}
For each position \(\mathfrak p_{i,t}\), we use the same balance-sheet convention as the full model: trader assets are posted collateral and trader liabilities are the negative of PNL. Formally,
\(A^{\mathrm{tr}}_{i,t}=c_{i,t}\) and \(L^{\mathrm{tr}}_{i,t}=-\mathrm{PNL}_{i,t}\),
so trader equity is
\(e_{i,t}=A^{\mathrm{tr}}_{i,t}-L^{\mathrm{tr}}_{i,t}=c_{i,t}+\mathrm{PNL}_{i,t}\).
This matches the standard interpretation: if \(\mathrm{PNL}_{i,t}<0\), the account owes the venue; if \(\mathrm{PNL}_{i,t}>0\), the venue owes the account.

At the venue level, exchange assets and liabilities are aggregates of trader-side quantities:
\(A_t^{x}=\sum_{i\in\mathcal I_t}A_{i,t}^{\mathrm{tr}}\) and \(L_t^{x}=\sum_{i\in\mathcal I_t}L_{i,t}^{\mathrm{tr}}\).
Therefore exchange solvency is equivalent to
\begin{equation}
A_t^{x}\ge L_t^{x}
\quad\Longleftrightarrow\quad
\mathsf{Solv}_t(\mathcal P_t)=A_t^x-L_t^x\ge 0.
\end{equation}
Using the position-level equity definition above, this becomes
\begin{equation}
\mathsf{Solv}_t(\mathcal P_t)=\sum_{\mathfrak p\in \mathcal P_t} e_t(\mathfrak p)
=\sum_{i\in\mathcal I_t} e_{i,t}.
\end{equation}
An exchange is insolvent when \(\mathsf{Solv}_t(\mathcal P_t)\le 0\). Equivalently, the exchange-level residual shortfall is
So \(S_t=\big(-\mathsf{Solv}_t(\mathcal P_t)\big)_+\).
When \(S_t>0\), the venue must either socialize losses (ADL), inject external capital, or remain undercollateralized.

\subsection{Liquidations}
When a user's equity approaches zero, the venue liquidates some or all of the position.
We denote by $\Delta q_{i, t}$ the amount of trader $i$'s position that is liquidated at time $t$, with $\Delta q_{i,t} \leq q_{i,t}$.
For this paper, only three prices matter:
\begin{itemize}
  \item \(p^{\mathrm{bk}}(\mathfrak p_{i,t})\): bankruptcy price, defined as the mark price $m_t$ such that \(e_{i,t}=0\).
  \item \(p^{\mathrm{liq}}(\mathfrak p_{i,t})\): liquidation trigger price, typically tied to maintenance margin.
  \item \(p^{\mathrm{liq,exec}}(\mathfrak p_{i,t},\Delta q_{i,t})\): realized order-book execution price for the liquidation trade.
\end{itemize}
The precise algorithms for liquidation vary from venue to venue.
These details are elided here as we focus on what happens when these algorithms fail.
We say that a liquidation fails and creates bad debt when execution is worse than bankruptcy:
\begin{itemize}
  \item Long closeout: shortfall if \(p^{\mathrm{liq,exec}}<p^{\mathrm{bk}}\).
  \item Short closeout: shortfall if \(p^{\mathrm{liq,exec}}>p^{\mathrm{bk}}\).
\end{itemize}
We note that liquidations may be full (\(\Delta q_{i,t}=q_{i,t}\)) or partial (\(\Delta q_{i,t}<q_{i,t}\)).

\subsection{Exchange risk management}
In practice, solvency enforcement is complex because the exchange generally cannot always liquidate positions at the mark price $m_t$.
For instance, if the order book on an exchange has less liquidity than $\Delta q_{i,t}$, then the liquidation cannot be fully executed.
The transaction costs of risk reduction are nondeterministic and depend on factors such as liquidity depth, market impact, order book elasticity, and the execution strategy employed.

When approaching the solvency boundary (i.e. when bad debt occurs), the key region is where positions appear solvent at mark price but not at expected liquidation execution price.
A standard response is partial liquidation up to the bankruptcy boundary, with ADL used for any remaining deficit.

\paragraph{Why real-time deterministic socialization is common in crypto venues.}
Compared with traditional clearing ecosystems, crypto perpetual venues operate with 24/7 markets, global pseudonymous membership, rapid withdrawal optionality, and limited ex post legal netting capacity.
When a stress event occurs, delayed discretionary recovery can trigger immediate run risk because users can attempt to exit collateral before losses are allocated.
Deterministic rule-based ADL therefore acts as an operational commitment device: losses are allocated at event time under public rules instead of relying on subsequent legal reconciliation.

\subsection{Positions and exchange position set}
An exchange is defined by a set of \emph{positions} that represent user accounts.
We define a per user account position object, \(\mathfrak{p}_{i,t}=\big(q_{i,t},\bar p_{i,t},c_{i,t},\mathrm{PNL}_{i,t},e_{i,t}\big)\):
where \(q_{i,t}\) is signed contract size, \(\bar p_{i,t}\) is average entry basis, \(c_{i,t}\) is posted collateral, \(\mathrm{PNL}_{i,t}\) is mark-to-market profit and loss, and \(e_{i,t}\) is account equity.
The set of active positions at round \(t\) is \(\mathcal{P}_t=\{\mathfrak{p}_{i,t}: i\in\mathcal{I}_t\}\).

\subsection{Insolvency and deficits}
Let \(L_t\subseteq\mathcal{I}_t\) be accounts with negative expected equity following liquidation.
The deficit at round $t$ is
\begin{equation}
D_t=\sum_{j\in L_t}(-e_{j,t}(p^{\mathrm{liq,exec}}(\mathfrak{p}_{j, t}, q_{j,t})))_+.
\end{equation}
We abuse notation and refer to $e_{j,t}(p^{\mathrm{liq,exec}}(\mathfrak{p}_{j, t}, q_{j,t}))$ as the trader's equity if their PNL was computed with the mark price $m_t = p^{\mathrm{liq,exec}}$.
The deficit represents the total negative equity in the system and represents the notional value of positions that need to be closed to preserve solvency.

\noindent Next, define the winner set \(W_t=\{i\in\mathcal{I}_t:\mathfrak{p}_{i,t}\in\mathcal{P}_t,\ \mathrm{PNL}_{i,t}>0\}\), with PNL haircut capacity \(u_{i,t}=(\mathrm{PNL}_{i,t})_+\) and \(U_t=\sum_{i\in W_t}u_{i,t}\).
The capacity $U_t$ is the maximum amount of profit that can be haircut to cover the deficit $D_t$.
This paper uses strict PNL-only haircuts: ADL reallocates positive PNL and does not haircut protected collateral.
In particular, each account-level seizure satisfies \(0\le x_{i,t}\le u_{i,t}\), so aggregate budget can only come from positive PNL capacity.
If \(D_t>0\) and $S_t < 0$, liquidation alone cannot restore solvency and the venue must socialize losses, inject outside capital, or remain undercollateralized.

\section{Autodeleveraging}\label{sec:adl}
We formulate autodeleveraging (ADL) as a sequential online decision problem: each round reveals a state, the venue chooses a feasible action, losses are realized, and the system transitions under unknown dynamics.
This follows canonical online learning and reinforcement learning formulations~\citep{Zinkevich2003,Hazan2016,SuttonBarto2018}.

\subsection{ADL Rounds}
Operationally, upon realizing a reference price $s_t$ that causes $D_t > 0$, the venue executes a sequence of operations: attempt liquidation, absorb residual bad debt with insurance when available, and finally, invoke ADL on winners if deficits remain.
In this paper, a \emph{round} is a defined as a contiguous sequence of order book fills created by an ADL algorithm.
That is, a round consists of the set of negative equity positions that are closed by matching with order book positions under a single execution of an ADL algorithm.
While ADL was initially constructed on the assumption that it would be used for a single round, October 10, 2025 demonstrated that this is not true in extreme stress scenarios.

\subsection{Lifecycle of an ADL Round}
A single ADL round has an execution lifecycle that separates severity choice (i.e. how much of the negative equity should be socialized) from allocation and matching (i.e. who should cover the negative equity positions).
First, the venue measures residual loser-side deficit \(D_t\) and chooses a budget to socialized, \(B_t\) (or equivalently chooses \(\theta_t \in (0,1)\) and sets $B_t = \theta_t D_t$).
Determining severity is the core of the ADL problem because required ADL transfer sizes \(q_k^{\mathrm{ADL}}\) are endogenous to liquidation outcomes: negative equity positions that are closed cross the order book at uncertain liquidation prices \(p_k^{\mathrm{liq,exec}}\).
These prices are unknown ex ante and must be estimated in live ADL policies.
Second, the ADL policy chooses winner-side reductions \(x_{i,t}\) (e.g. queue, pro-rata, or another policy).
Third, ADL matches winners (positive equity positions) and losers (negative equity positions to be closed) at the bankruptcy transfer price \(p_k^{\mathrm{bk}}\) of the loser (deterministic given the marked state).
Figure~\ref{fig:adl-lifecycle} summarizes this lifecycle and where the benchmark \(B_t^{\mathrm{needed}}\) is measured.

\subsection{Haircut Benchmark}
For an ADL transfer \(k\) in round \(t\), let \(p_k^{\mathrm{liq,exec}}\) be realized liquidation execution price, \(p_k^{\mathrm{bk}}\) the bankruptcy transfer price used by ADL matching, and \(q_k^{\mathrm{ADL}}\) the signed ADL transfer quantity.
In live operation, liquidation execution prices are not known ex ante, so the venue uses an estimator \(\hat p_k^{\mathrm{liq,exec}}\) and an implied estimated quantity \(\hat q_k^{\mathrm{ADL}}\).
These quantities are used to estimate what fraction of the deficit will be covered by ADL (i.e. severity).
A natural benchmark for estimating the shortfall covered is the expected budget to socialize via ADL:
\begin{equation}
\widehat B_t^{\mathrm{needed}}
=
\sum_{k\in t}\left|\hat p_k^{\mathrm{liq,exec}}-p_k^{\mathrm{bk}}\right|\,|\hat q_k^{\mathrm{ADL}}|.
\end{equation}
This can be interpreted as the ADL algorithm's ex ante estimate of the shortfall covered in one ADL round.
Note that because it can only realize $p^{\mathrm{liq, exec}}$ ex post, we also need to compare to the realized, ex post benchmark:
\begin{equation}
B_t^{\mathrm{needed}}
=
\sum_{k\in t}\left|p_k^{\mathrm{liq,exec}}-p_k^{\mathrm{bk}}\right|\,|q_k^{\mathrm{ADL}}|.
\end{equation}
But note that live, online ADL algorithms cannot directly use $B_t^{\mathrm{needed}}$; it will only be a benchmark for theoretical analysis.

\paragraph{Units and interpretation.}
Most quantities in this paper are dollar-valued (for example \(D_t\), \(B_t^{\mathrm{needed}}\), and \(H_t\)), so reported losses and totals are in dollars.
The main unitless quantity is severity \(\theta_t\in[0,1]\), defined by \(B_t=\theta_t D_t\), which is the fraction of the observed deficit socialized in round \(t\).
Over- and under-socialization are then dollar gaps between \(H_t\) and \(B_t^{\mathrm{needed}}\).

\subsection{State, action, and policy spaces}
Let rounds be indexed by \(t=1,\dots,T\).
Using the position-set notation from~\S\ref{sec:background}, define state space \(\mathcal{S}\) and action space \(\mathcal{A}\) as follows.

\paragraph{State.}
Round state is \(s_t=\big(\mathcal{P}_t,D_t,W_t,u_t,\zeta_t\big)\in\mathcal{S}\), where \(\mathcal{P}_t\) is the exchange position set, \(D_t\) is residual deficit, \(W_t\) is the winner index set, and \(u_t=(u_{i,t})_{i\in W_t}\) are winner capacities.
Here \(\zeta_t\) is an auxiliary vector of round-start, policy-observable market signals (e.g., price and volatility snapshots, spread/depth summaries, and recent liquidation-flow aggregates).

\paragraph{Action.}
An ADL action is \(a_t=(B_t,x_t)\in\mathcal{A}(s_t)\), with aggregate budget \(B_t\) and allocation vector \(x_t=(x_{i,t})_{i\in W_t}\).
The feasible action set is
\begin{equation}
\mathcal{A}(s_t)=
\left\{(B,x):\ 0\le B\le U_t,\ 0\le x_i\le u_{i,t},\ \sum_{i\in W_t}x_i=B\right\}.
\end{equation}
Equivalent parameterization uses severity \(\theta_t\in[0,1]\) and haircut fractions \(h_{i,t}\in[0,1]\): \(B_t=\theta_t D_t,\ x_{i,t}=h_{i,t}(u_{i,t}+\varepsilon)\).
Here \(\varepsilon>0\) is a small numerical regularizer (in the same units as \(u_{i,t}\) and \(D_t\)) used only to avoid divide-by-zero or near-zero denominators in normalized ratios (e.g., \(x_{i,t}/(u_{i,t}+\varepsilon)\) and \(\theta_t^{\mathrm{needed}}=B_t^{\mathrm{needed}}/(D_t+\varepsilon)\)).

\paragraph{Policy.}
An ADL policy is a history-dependent map \(\pi_t:\mathcal{H}_t\to\mathcal{A}(s_t)\), with \(\mathcal{H}_t=(s_1,a_1,\dots,s_{t-1},a_{t-1},s_t)\).
State transitions follow unknown stochastic dynamics \(s_{t+1}=F_t(s_t,a_t,\omega_t)\) with shock \(\omega_t\).

\subsection{Canonical policy examples}
There are two high-level ADL policy families used in practice and analysis: queue-based policies and partial haircut policies.

\paragraph{Queue-based Policies.}
Queueing policies first choose a score \(s_{i,t}\) for each winner, then order positions by score, then close positions in that order.
Formally, if \(\sigma_t\) is a permutation that orders the scores, i.e.
\(
s_{\sigma_t(1),t}\ge\cdots\ge s_{\sigma_t(|W_t|),t},
\)
then queue allocation greedily fills budget \(B_t\): top-ranked accounts are fully closed, with \(x_{\sigma_t(j),t}=u_{\sigma_t(j),t}\).
That is, if $i^{\star} = \min \{ j : \sum_{i=1}^j x_{\sigma_t(i), t} \geq B_t \}$, then positions $\sigma_t(1), \ldots, \sigma_t(i^*-1)$ are fully closed while $\sigma_t(i^*)$ is partially closed.
This structure can produce \(100\%\) haircuts for early queue positions.

\paragraph{Partial Haircut Policies.}
Partial haircut policies distribute round budget across winners so no touched account is fully closed during ADL.
Writing \(h_{i,t}=x_{i,t}/(u_{i,t}+\varepsilon)\), these policies enforce \(h_{i,t}<1\) for participating winners (equivalently, a hard cap \(h_{i,t}\le \bar h<1\)).
This avoids queue-style full closeout of a position and smooths the burden across accounts due to the budget constraint. \\

\noindent \emph{Pro-Rata Policy.}
The quintessential partial haircut policy is pro-rata.
This policy allocates proportionally to positive-PNL capacity: \(x^{\mathrm{PR}}_{i,t} = \frac{u_{i,t}}{U_t}\,B_t\).
This equalizes haircut fraction \(x_{i,t}/u_{i,t}\) across winners and is the canonical fairness benchmark.
In the full model, this benchmark is additionally motivated by axiomatic fairness: under Sybil resistance, scale invariance, and monotonicity, capped pro-rata is the unique stable allocation family \citep{ChitraEtAl2026ADLPreprint}. 
Moreover, multiple live protocols utilize this rule in practice (e.g. Drift and Paradex)~\citep{ChitraEtAl2026ADLPreprint}.\\

\noindent \emph{Min-max Integer Linear Program (ILP).}
In practice, perpetuals exchanges represent contracts (\eg~user positions) as integer counts.
For instance, a single contract might represent a position of size 0.1 BTC.
To ensure that such integrality constraints are met, the pro-rata policy is modified to a rounding integer linear program.
This program minimizes worst-account burden subject to exact budget and lot-feasible execution:
\begin{equation}\label{eq:ilp}
\min_{x,z}\ z
\quad\text{s.t.}\quad
\sum_{i\in W_t}x_i=B_t,\ \
0\le x_i\le u_{i,t},\ \
\frac{x_i}{u_{i,t}+\varepsilon}\le z,\ \
x_i\in\mathcal G_{i,t},
\end{equation}
where \(\mathcal G_{i,t}\) is the contract integral constant for account \(i\).
We note that in~\cite{ChitraEtAl2026ADLPreprint}, it was shown that the empirical rounding error from using the ILP (versus pro-rata) can add up to millions of dollars in practice.

\subsection{Per-round optimization objective}
Each round requires two linked choices: severity (how much budget \(B_t\) to raise) and allocation (how to split \(B_t\) across winners).
As such, objective should penalize both solvency-tracking error and how concentrated haircuts are.
A natural asymmetric round loss is
\begin{equation}\label{eq:loss}
\tilde\ell_t(x_t)=
\lambda_{\mathrm{under}}[B_t^{\mathrm{needed}}-H_t]_+
+\lambda_{\mathrm{over}}[H_t-B_t^{\mathrm{needed}}]_+
+\lambda_{\mathrm{fair}}\,\Gamma_t(x_t,u_t),
\end{equation}
where \(H_t=\sum_i x_{i,t}\) and \(\Gamma_t\) is a concentration functional.
The undershoot and overshoot terms encode solvency-tracking asymmetry, while \(\Gamma_t\) captures how fairly losses are divided amongst users.

For online control and empirical comparability under partial observability, we deploy the convex surrogate
\begin{equation}
\ell_t(x_t)=
\lambda_{\mathrm{track}}\left|H_t-B_t^{\mathrm{needed}}\right|
+\lambda_{\mathrm{fair}}\max_{i\in W_t}\frac{x_{i,t}}{u_{i,t}+\varepsilon}.
\end{equation}
The fairness proxy is the largest fraction of available positive PNL seized from any winner, which matches the burden-concentration diagnostic used in policy evaluation.
Equivalently, writing \(h_{i,t}=x_{i,t}/(u_{i,t}+\varepsilon)\), the objective penalizes \(\max_i h_{i,t}\): this is a worst-hit-participant criterion and corresponds to the min--max burden objective used in our ILP benchmark (Appendix~\ref{app:duality-queue-position}).
As both terms are convex on the feasible polytope (see Appendix~\ref{app:duality-queue-position}), we can solve
\begin{equation}
x_t\in\arg\min_{x\in\mathcal{X}_t}\ell_t(x),
\qquad
\mathcal{X}_t=\{x:\exists B_t\ \text{s.t.}\ (B_t,x)\in\mathcal{A}(s_t)\}.
\end{equation}
\emph{Tracking-versus-allocation decomposition.}
Under exact budget execution, feasibility enforces \(H_t=\sum_i x_{i,t}=B_t\), so \(\left|H_t-B_t^{\mathrm{needed}}\right|=\left|B_t-B_t^{\mathrm{needed}}\right|\).
However, this implies that if there is positive tracking error, then it is because the severity (\ie~the budget) was chosen incorrectly.
We note that the tracking error is primarily driven by the choice of severity, whereas queue/pro-rata choice primarily affects burden allocation (the fairness term).
In practice, additional tracking error can still arise from integral constraints, partial fills/latency, and ex-ante execution-price estimation error.

\subsection{Static and Dynamic Regret}
We use standard online-learning definitions of static and dynamic regret~\citep{Zinkevich2003,Hazan2016,ShalevShwartzBenDavid2014,Besbes2015}.
We connect these to instance-dependent guarantees that scale with realized gradient magnitude or loss curvature rather than horizon alone~\citep{Duchi2011,Gaillard2014SecondOrder}.
For a policy sequence \(\pi=\{x_t\}_{t=1}^T\), static regret is
\begin{equation}
\mathrm{Reg}^{\mathrm{static}}_T(\pi)=
\sum_{t=1}^T\ell_t(x_t)-\min_{x\in\mathcal{X}}\sum_{t=1}^T\ell_t(x),
\end{equation}
and dynamic regret against comparator sequence \(\{x_t^\star\}\) is
\begin{equation}
\mathrm{Reg}^{\mathrm{dyn}}_T(\pi)=
\sum_{t=1}^T\ell_t(x_t)-\sum_{t=1}^T\ell_t(x_t^\star).
\end{equation}
For empirical claims, we make the comparator class explicit.
Let \(\mathcal P\) be the implementable policy library used in replay.
Policy-class regret is
\begin{equation}
\mathrm{Reg}^{\mathcal P}_T(\pi)=
\sum_{t=1}^T\ell_t(x_t^\pi)-\min_{\pi'\in\mathcal P}\sum_{t=1}^T\ell_t(x_t^{\pi'}).
\end{equation}
We also isolate scalar tracking regret:
\begin{equation}
\mathrm{Reg}^{\mathrm{track}}_T(\pi)=
\sum_{t=1}^T\left|H_t^\pi-B_t^{\mathrm{needed}}\right|
-\min_{\pi'\in\mathcal P}\sum_{t=1}^T\left|H_t^{\pi'}-B_t^{\mathrm{needed}}\right|.
\end{equation}
Terminology is fixed throughout:
\emph{tracking error} means \(\left|H_t-B_t^{\mathrm{needed}}\right|\),
\emph{overshoot} means \([H_t-B_t^{\mathrm{needed}}]_+\),
and \emph{undershoot} means \([B_t^{\mathrm{needed}}-H_t]_+\).
Finally, for objective diagnostics we use
\begin{equation}\label{eq:regret-decomp}
\mathcal{R}_{t}^{\mathrm{track}}=\lambda_{\mathrm{track}}\left|H_t-B_t^{\mathrm{needed}}\right|,\quad
\mathcal{R}_{t}^{\mathrm{fairness}}=\lambda_{\mathrm{fair}}\max_{i\in W_t}\frac{x_{i,t}}{u_{i,t}+\varepsilon},\quad
\mathcal{R}_{t}^{\mathrm{total}}=\mathcal{R}_{t}^{\mathrm{track}}+\mathcal{R}_{t}^{\mathrm{fairness}}.
\end{equation}

\subsection{Basic regret bound for the objective}
For the objective~\eqref{eq:loss} in this paper, the tightest and most interpretable bound is the one-dimensional severity result below.
This result upper bounds the regret in terms of the deficit and the variation of the severity.

\begin{proposition}[Deficit-weighted severity bound with explicit constants]\label{prop:severity-explicit}
Consider the one-dimensional severity parameterization \(B_t=\theta_t D_t\), \(\theta_t\in[0,1]\), with loss $\ell_t^{\theta}(\theta_t)=D_t\left|\theta_t-\theta_t^{\mathrm{needed}}\right|,
, \theta_t^{\mathrm{needed}}=\min\!\left\{1,\frac{B_t^{\mathrm{needed}}}{D_t+\varepsilon}\right\}$
Let comparator path variation be $P_T^{\theta}=\sum_{t=2}^{T}\left|\theta_t^\star-\theta_{t-1}^\star\right|$.
Projected OGD on \([0,1]\) with step \(\eta>0\) satisfies
$\mathrm{Reg}^{\mathrm{dyn},\theta}_T
\le
\frac{1+2P_T^{\theta}}{2\eta}
\;+\;
\frac{\eta}{2}\sum_{t=1}^{T}D_t^2,$
since \(|\partial \ell_t^\theta|\le D_t\).
With
\(
\eta^\star=\sqrt{(1+2P_T^\theta)/\sum_t D_t^2}
\),
\begin{equation}
\mathrm{Reg}^{\mathrm{dyn},\theta}_T
\le
\sqrt{(1+2P_T^\theta)\sum_{t=1}^{T}D_t^2}.
\end{equation}
\end{proposition}
This bound has an intuitive two-factor form.
\(\sum_t D_t^2\) captures event scale: larger deficits make the episode harder to reduce for any policy.
\(P_T^\theta\) captures how quickly target severity moves across rounds.
In ADL, this variation is larger when exchanges mispredict ex post liquidation execution prices (forcing repeated severity corrections) and when closeout liquidity is poor (execution outcomes become more unstable round to round).
So the same mechanism can have low regret in calm, predictable episodes and much higher regret in bursty, illiquid cascades even at similar notional scale.
Relative to prior ADL analysis in~\citep{ChitraEtAl2026ADLPreprint}, this explicit \(P_T^\theta\) term makes the execution and liquidity costs visible in the bound itself.
In this sense, the bound is instance-calibrated: since the per-round slope is \(|\partial \ell_t^\theta|\le D_t\), the rate depends on \(\sqrt{\sum_t D_t^2}\), analogous to adaptive/second-order OCO bounds~\citep{Duchi2011,Gaillard2014SecondOrder}.
This is exactly the lens used in~\S\ref{sec:empirical} to interpret October 10.
Full proofs are in Appendix~\ref{app:regret-bounds-proofs}.
\section{Regret Minimization}
This section analyzes regret minimization for the ADL objective~\eqref{eq:loss}.
We show why queue-based policies are structurally high-regret, characterize trade-offs among no-regret controllers, and connect these predictions to empirical diagnostics.

\subsection{Why queue policies are structurally high-regret}
Queue mechanisms impose a fixed ranking map and then allocate budget greedily.
In nonstationary episodes, ADL regret has two main channels: \emph{tracking error} from severity mismatch (\(B_t\) vs.\ \(B_t^{\mathrm{needed}}\)) and \emph{fairness error} from burden concentration.
The fixed queue map mainly drives the fairness error by repeatedly overloading top-ranked winners, while severity mistakes are amplified by rounding, partial fills, and execution uncertainty.
Hence queue policies are expected to accumulate regret quickly in clustered-stress episodes, as in October 10, 2025.
We first illustrate this explicitly with a single example.

\paragraph{Fixed queue policies can incur linear regret.}\label{ex:queue-linear-regret}
Consider two winning accounts \(i\in\{1,2\}\), an exact round budget \(B_t=1\), a feasible set
$\mathcal{X}_t=\{x\in\mathbb{R}_+^2:\ x_1+x_2=1,\ 0\le x_i\le u_{i,t}\}$,
and consider the fairness-only loss $\ell_t(x)=\lambda_{\mathrm{fair}}\max_i \frac{x_i}{u_{i,t}}$.
Let capacities alternate with \(M>1\):
\[
(u_{1,t},u_{2,t})=
\begin{cases}
(1,M), & t\ \text{odd},\\
(M,1), & t\ \text{even}.
\end{cases}
\]
Now impose a fixed queue that always serves account \(1\) first and greedily fills the full budget before considering account \(2\). Since \(u_{1,t}\ge 1\) in every round, this queue always outputs $x_t^{\mathrm{queue}}=(1,0)$.
Its per-round loss is then:
\[
\ell_t(x_t^{\mathrm{queue}})=
\begin{cases}
\lambda_{\mathrm{fair}}\max\{1,\ 0/M\}=\lambda_{\mathrm{fair}}, & t\ \text{odd},\\
\lambda_{\mathrm{fair}}\max\{1/M,\ 0\}=\lambda_{\mathrm{fair}}/M, & t\ \text{even}.
\end{cases}
\]
A dynamic comparator that allocates to the higher-capacity account each round chooses
\[
x_t^\star=
\begin{cases}
(0,1), & t\ \text{odd},\\
(1,0), & t\ \text{even},
\end{cases}
\]
which yields \(\ell_t(x_t^\star)=\lambda_{\mathrm{fair}}/M\) in every round. Therefore, for even \(T\),
\begin{equation}
\label{eq:queue-linear-omega}
\mathrm{Reg}^{\mathrm{dyn}}_T(\mathrm{queue})
=\sum_{t=1}^{T}\left[\ell_t(x_t^{\mathrm{queue}})-\ell_t(x_t^\star)\right]
=\frac{T}{2}\lambda_{\mathrm{fair}}\!\left(1-\frac{1}{M}\right)=\Omega(T).
\end{equation}
So even with exact budget tracking (\(x_{1,t}+x_{2,t}=B_t\) each round), a fixed queue can accumulate linear fairness regret under alternating round geometry.

\subsection{Estimation of $p^{\mathrm{exec}}$: tracking failure beyond regret}
\label{sec:pexec-estimation}

In practice, regret is a coarse measurement that doesn't account for how execution costs impact tracking error.
Execution costs can make realized severity too low to cover deficits.
In this section, we adapt our regret bounds to include execution-based failure modes.

\paragraph{Why execution-price estimation enters tracking.}
As feasibility enforces \(H_t=\sum_{i\in W_t}x_{i,t}=B_t\), within-round allocation does not by itself create tracking error.
Tracking error is mainly due to the choice of severity $\theta_t$, which chooses $B_t = \theta_t D_t$.
Severity depends on ADL quantities \(q_k^{\mathrm{ADL}}\), which are only known ex post because they are induced by uncertain liquidation execution at price \(p_k^{\mathrm{liq,exec}}\).

\paragraph{Estimated versus ex post benchmarks.}
Let \(\Delta p_{k,t}(q) := |p^{\mathrm{exec,liq}}_{k,t}(q) - p^{\mathrm{bk}}_{k,t}|\), where \(p^{\mathrm{bk}}\) is the bankruptcy transfer price
used for matching. The ex post (policy-comparable) needed benchmark is
\begin{equation}
\label{eq:needed-expost}
B_t^{\mathrm{needed}} := \sum_{k\in t} \Delta p_{k,t}(q^{ADL}_{k,t})\, |q_{k,t}^{\mathrm{ADL}}|.
\end{equation}
In deployment, the venue does not know \(q_{k,t}^{\mathrm{ADL}}\) when choosing severity and therefore uses \(\hat q_{k,t}^{\mathrm{ADL}}\) from \(\hat p^{\mathrm{liq,exec}}_{k,t}(\cdot)\).
We analogously define \(\Delta \hat{p}_{k,t}(q):=|\hat p^{\mathrm{exec,liq}}_{k,t}(q)-p^{\mathrm{bk}}_{k,t}|\) and
\begin{equation}
\label{eq:needed-exante}
\widehat B_t^{\mathrm{needed}} := \sum_{k\in t} \Delta \hat{p}_{k,t}(\hat q^{ADL}_{k,t})\, |\hat q_{k,t}^{\mathrm{ADL}}|.
\end{equation}
An ADL mechanism aims to target \(H_t\approx \widehat B_t^{\mathrm{needed}}\) to reduce severity tracking error.

\paragraph{What is new in this separation.}
Relative to~\citep{ChitraEtAl2026ADLPreprint}, we explicitly separate the ex ante quantity used to make decisions (\(\widehat B_t^{\mathrm{needed}}\)) from the ex post replay benchmark used for evaluation (\(B_t^{\mathrm{needed}}\)).
This yields a sharper interpretation: regret against the ex ante target measures online control quality, while the ex ante/ex post gap isolates execution-model error and liquidity-driven mismatch.
This is intended as a direct correction to known measurement ambiguities discussed in public critiques~\citep{Thogiti2025ADL653M,Thogiti2026ADLExplainer} of the first paper to formalize ADL~\cite{ChitraEtAl2026ADLPreprint}.

\paragraph{Ex post severity error.}
When execution is underestimated, \(H_t\) can fall short of \(B_t^{\mathrm{needed}}\), so we track this gap directly via
\begin{equation}
\label{eq:VT-failure}
V_T \;:=\; \sum_{t=1}^T \big[B_t^{\mathrm{needed}} - H_t\big]_+,
\end{equation}
which can be large even when regret is small.

\paragraph{Total Tracking Error Minimization.}
We now consider a notion of total tracking error, which is the sum of regret and ex post severity error.
Let \(\ell_t(x;b)\) denote the deployed convex surrogate with tracking target \(b\):
\[
\ell_t(x;b) := \lambda_{\mathrm{track}}\,\big|{\bf 1}^\top x - b\big|
+ \lambda_{\mathrm{fair}}\max_{i\in W_t}\frac{x_{i,t}}{u_{i,t}+\varepsilon}.
\]
In Appendix~\ref{app:pexec-estimation}, we show that for any policy sequence \(\pi=\{x_t\}\) and comparator class \(\mathcal{P}\):
\begin{equation}
\label{eq:total-error-decomp}
\sum_{t=1}^T \ell_t(x_t)
\;\le\;
\min_{\pi'\in\mathcal{P}}\sum_{t=1}^T \ell_t(x^{\pi'}_t)
\;+\;
\underbrace{\mathrm{Reg}^{\mathcal{P}}_{T}(\pi;\hat\ell)}_{\text{online control / optimization error}}
\;+\;
\underbrace{2\lambda_{\mathrm{track}}\sum_{t=1}^T \big|B_t^{\mathrm{needed}}-\widehat B_t^{\mathrm{needed}}\big|}_{\text{execution-estimation error}}.
\end{equation}
When \(H_t\approx \widehat B_t^{\mathrm{needed}}\), ex post severity error scales with optimistic estimation error:
\(V_T\lesssim\sum_{t=1}^T [B_t^{\mathrm{needed}}-\widehat B_t^{\mathrm{needed}}]_+\).
Appendix~\ref{app:pexec-estimation} gives the formal statement and proof of~\eqref{eq:total-error-decomp}.

\paragraph{Linear Impact Model.}
To illustrate how queues can underperform other ADL mechanisms, we consider a simple linear price impact model for execution cost.
This model is stylized, but it clearly shows how estimation error can shift performance from sublinear to linear.

Our linear model assumes that at each time $t$, \(p^{\mathrm{liq,exec}}_t(q)=p^{\mathrm{mark}}\mp \alpha_t q\), for an unknown local slope \(\alpha_t\).
Note that the sign of $\alpha_t$ is direction-dependent: use \(-\) for liquidating long inventory (sell pressure lowers execution price) and \(+\) for liquidating short inventory (buy pressure raises execution price).
We assume that ex ante, prior to executing an ADL mechanism, the designer estimates a slope $\hat{\alpha}_t$ and sets $\hat{p}^{\mathrm{exec,liq}}_t(q) = p^{\mathrm{mark}}_t \mp \hat{\alpha}_t q$.
Using equations~\eqref{eq:needed-exante} and~\eqref{eq:needed-expost}, this implies that
\[
B_t^{\mathrm{needed}} - \hat{B}_t^{\mathrm{needed}} \leq \sum_{k \in t} |p^{\mathrm{exec,liq}}_{k,t}(q_{k,t}^{ADL})-\hat{p}^{\mathrm{exec,liq}}_{k,t}(q_{k,t}^{ADL})| |q_{k,t}^{ADL}| = |\alpha_t - \hat{\alpha}_t| \left(\sum_{k \in t} |q_{k,t}^{ADL}|^2\right)
\]
If for some \(Q > 0\), we have \(q_{k,t}^{ADL} \leq Q\), then \(\left|B_t^{\mathrm{needed}}-\hat{B}_t^{\mathrm{needed}}\right| \leq Q^2 |\alpha_t - \hat{\alpha}_t|\) round-wise.
This implies that, for example, if the mechanism estimates $\hat{\alpha}_t$ such that $|\hat{\alpha}_t - \alpha_t| = O(\frac{1}{\sqrt{T}})$, then $V_T = O(\sqrt{T})$.
On the other hand, if for some $C > 0$, $|\alpha_t - \hat{\alpha}_t| \geq C$, then $V_T = \Omega(T)$.
In Appendix~\ref{app:pexec-estimation}, we formalize this in Proposition~\ref{prop:VT-linear-model} and show that if the total estimation error $\sum_{t} |\alpha_t - \hat{\alpha}_t|$ is sufficiently smooth then we have sublinear regret.
Note that one can interpret this result as saying that an ADL mechanism can handle a constant number (in $T$) of large misestimations and still have low severity tracking error.

\subsection{Queue Instability}
Given tracking errors, a natural question to ask is how robust ADL mechanisms are to perturbations.
We study this theoretically using convex analysis (Appendix~\ref{app:queue-instability}) and a linear impact feedback model (Appendix~\ref{app:pexec-estimation}).
The main critique is not that queue allocation must always break immediate solvency restoration.
When budget execution is exact (\(H_t=B_t\)) and matching is fixed, immediate tracking is primarily a severity-choice question.
Queue failure instead appears through mechanism quality: discontinuity of the allocation map, rank instability under small perturbations, and strict dominance gaps under min--max burden objectives.
These are the channels used below and in Section~\ref{sec:empirical} to evaluate queue performance.

\paragraph{Linear Impact Model.}
Under the linear impact model, Appendix~\ref{app:pexec-estimation} (Proposition~\ref{prop:queue-churn-omegaT}) shows that a fixed queue ordering can create \(\Omega(T)\) variation in effective liquidation conditions, even when accounts do not return once closed.
On the same sequence, smooth mixing policies keep this variation bounded.
The effect is indirect: allocation changes which winners are closed each round, which changes next-round state composition and can change later liquidation slopes.
In Section~\ref{sec:empirical}, replay keeps the realized market path fixed, so we do not estimate the size of this feedback effect; we leave that quantification to future work.

\paragraph{Convex Analysis.}
In Appendix~\ref{app:queue-instability}, we use convex geometry to demonstrate a similar lack of queue robustness.
We define a feasibility polytope that represents the set of feasible severities and allocations that can be chosen.
We show that queue mechanisms correspond to extreme points of this polytope.
This implies, in particular, that arbitrarily small changes in scores can lead to large changes in ADL allocation. 
The empirical instability diagnostics in Section~\ref{sec:empirical} (monotonicity violations and adjacent-round rank stability) are intended to measure this convex geometric instability.

\section{Empirical Analysis}\label{sec:empirical}
We study the October 10, 2025 Hyperliquid stress event using public replay artifacts from~\citep{HyperMultiAssetedADL}, canonical event data from~\citep{HyperReplayCanonicalData, PluriholonomicADLRepo}.
The 21:16--21:27 UTC window contains about \$2.103B in liquidations and decomposes into \(T=16\) ADL rounds using the methodology of~\cite{ChitraEtAl2026ADLPreprint}.
This clustered-round structure is the empirical motivation for online evaluation.
Unless stated otherwise, empirical counts are reported at the canonical ADL-event level rather than per-user or per-fill derived rows.

\paragraph{Observation model and replay invariants.}
Claims are made under a public-data observation model, not full internal-ledger reconstruction.
Across counterfactual policies, replay holds fixed round boundaries, loser deficits \(D_t\), winner sets \(W_t\), winner capacities \(u_{i,t}\), the observable round-start context \(\zeta_t\), bankruptcy transfer prices \(p_k^{\mathrm{bk}}\), and the realized market price path.
Policy differences therefore enter through severity and allocation, not through re-labeling the round state.

\paragraph{What is not modeled.}
We do not model policy feedback into liquidation behavior, order-book resilience, HLP/market-maker participation, strategic order placement, or withdrawals during the event.
Results should therefore be interpreted as partial-equilibrium policy comparisons on a fixed realized path, not as a full market-equilibrium simulation.

\paragraph{Markout horizon and economic interpretation.}
As in~\cite{ChitraEtAl2026ADLPreprint}, we use markout to estimate the counterfactual value of a position closed by ADL.
The horizon parameter \(\Delta\) is swept from immediate to short-horizon evaluation windows.
Economically, this captures plausible unwind-value uncertainty in stressed books rather than a single point estimate.

\paragraph{Policy classes and information sets.}
We evaluate production Hyperliquid queue, integer pro-rata, continuous pro-rata, vector mirror descent, and min-max ILP under the same replay rounds.
Deployable policies (production queue, integer pro-rata, vector mirror descent) use only round-start inputs.
Ex post reference policies (continuous pro-rata, min-max ILP) use realized replay quantities unavailable at decision time and are reported only as reference lower bounds.

For policy \(\pi\), define the round-level concentration ratio \(m_t^\pi := \max_{i\in W_t}\frac{x_{i,t}^{\pi}}{u_{i,t}+\varepsilon}\).
In the empirical decomposition, we report
\begin{equation}
\ell_{t,\mathrm{emp}}^{\lambda}(\pi)
=
\left|H_t^{\pi}-B_t^{\mathrm{needed}}\right|
\;+\;
\lambda\,B_t^{\mathrm{needed}}\left|m_t^\pi-m_t^{\mathrm{ILP}}\right|,
\end{equation}
where \(m_t^{\mathrm{ILP}}\) is the min--max ILP reference concentration on the same replay round.

\paragraph{Metric glossary.}
We report cumulative objective value \(L^\lambda(\pi)=\sum_t \ell_{t,\mathrm{emp}}^\lambda(\pi)\), its tracking and fairness components \(T(\pi)\) and \(F(\pi)\), and policy-class regret only when explicitly labeled.

\subsection{Empirical Objective and Regret Diagnostics}
The falsifiable prediction is that, on fixed replay rounds with fixed state trajectory, adaptive/smoother policies achieve lower cumulative tracking and fairness components than queue baselines.
We tested this on the October 10 replay data, and the results ranked policies as predicted, with Hyperliquid's production queue the worst performer (Table~\ref{tab:regret-totals}, Figures~\ref{fig:total-regret-short},~\ref{fig:overshoot-regret-short}, and~\ref{fig:fairness-short}).

\paragraph{Main findings.}
When \(\lambda=1\), production queue has \(T(\pi)=\$53.78\)M tracking component and \(F(\pi)=\$11.08\)M fairness component, for \(L^\lambda(\pi)=\$64.86\)M total (Table~\ref{tab:regret-totals}, Figure~\ref{fig:total-regret-short}).
At \(\Delta=0\), production overshoot \([H_t-B_t^{\mathrm{needed}}]_+\) is about \$45.0M; under short-horizon markout sweep it remains \$45.0M--\$51.7M (Figure~\ref{fig:horizon-short}).
Decompositions are shown in Figures~\ref{fig:per-round-short},~\ref{fig:overshoot-regret-short}, and~\ref{fig:cum-regret-short}.

\paragraph{Deployable subset (deployment-relevant).}
Within deployable policies, production cumulative objective value is \$64.86M, versus \$3.40M (integer pro-rata) and \$4.41M (vector mirror descent), a reduction of roughly one order of magnitude.
This ordering is visible in Table~\ref{tab:regret-totals} and Figures~\ref{fig:total-regret-short} and~\ref{fig:cum-regret-short}.
Ex post reference policies are lower still, but these require information that is unavailable when the decision is made.

\paragraph{Interpretation.}
The key empirical quantity is the gap between executed severity \(H_t\) and the replay transfer benchmark \(B_t^{\mathrm{needed}}\).
In this event, production queue both overshoots this benchmark and concentrates burden on a smaller set of winners than the best deployable alternatives (Figures~\ref{fig:horizon-short},~\ref{fig:cum-overshoot-short}, and~\ref{fig:fairness-short}).

\subsection{Queue Instability}
Queue robustness is evaluated with two diagnostics: monotonicity violations (adjacent inversions after ADL) and adjacent-round rank stability of normalized burden.
Results are consistent with Appendix~\ref{app:queue-instability}: production queue has much higher inversion rates (about \(11.4\%\)) than vector mirror descent (about \(0.64\%\)) and pro-rata (approximately \(0\%\)), and lower rank stability (about \(0.34\) vs \(0.72\) for vector and \(0.86\) for pro-rata).
Definitions and full plots are in Appendix~\ref{app:queue-instability} and Figures~\ref{fig:mono-viol-short} and~\ref{fig:rank-stability-short}.

\subsection{Instance-Calibrated Upper Envelope}\label{eq:emp-regret-bound}
From Proposition~\ref{prop:severity-explicit},
\(\mathrm{Reg}^{\mathrm{dyn},\theta}_T
\le
\sqrt{(1+2P_T^\theta)\sum_{t=1}^{T}D_t^2}.\)
Using replay-derived estimates, \(\widehat P_T^\theta=7.06\) for \(T=16\), define the instance-calibrated upper envelope
\[
\mathcal{B}_{\mathrm{inst}}
:=
\sqrt{(1+2\widehat P_T^\theta)\sum_{t=1}^{T}D_t^2}
\approx \$129.7\text{M}.
\]
This is not an information-theoretic worst-case statement; it is an episode-calibrated envelope obtained by plugging realized deficits and estimated variation into Proposition~\ref{prop:severity-explicit}.
Production cumulative objective value is \$64.86M, i.e. about \(50.0\%\) of \(\mathcal{B}_{\mathrm{inst}}\).
The best baseline that uses only start-of-round information has \$3.40M (about \(2.6\%\) of \(\mathcal{B}_{\mathrm{inst}}\)), so roughly \(47.4\%\) of \(\mathcal{B}_{\mathrm{inst}}\) reflects avoidable production loss on this path.
Figures~\ref{fig:regret-to-bound-ratio-short} and~\ref{fig:cum-bound-vs-regret-short} plot total objective against this bound.

\section{Conclusion}
This paper studied the theoretical and empirical properties of autodeleveraging policies under repeated use.
We formulated a compact online-learning model and used it to evaluate ADL performance on October 10, 2025.
A key clarification is the distinction between exchange deficit \(D_t\), replay transfer benchmark \(B_t^{\mathrm{needed}}\), and realized winner haircut \(H_t\).
Under this accounting, production queue exhibits large over-socialization relative to \(B_t^{\mathrm{needed}}\), while online-feasible alternatives materially reduce tracking component and total objective.
Replay-oracle baselines (i.e., policies that use ex post replay quantities unavailable at round start) are useful upper bounds, but are not interpreted as deployment-feasible event-time controllers.
Our queue-instability analysis gives a churn-robust structural critique: discontinuous extreme-point selection can induce high effective nonstationarity under repeated rounds, though the empirical magnitude of this feedback channel remains to be identified.
Our results complement those of~\cite{ChitraEtAl2026ADLPreprint}, in that we provide a positive result relative to that paper's negative result of the so-called ADL trilemma.

% This paper reformulates ADL as an online learning problem with a PNL-haircut state space.
% That reformulation yields a compact evaluation framework: budget-tracking regret, fairness regret, and total regret over stress rounds.

% In the October 10, 2025 case study, production queue ADL has high cumulative regret relative to transparent alternatives.
% The implication is practical: ADL mechanism choice materially changes trader harm and long-run venue quality, even when solvency targets are fixed.

% In this paper, we focus on the components directly estimable from public data (tracking and fairness) under an explicit observation model.
% The broader trilemma result from \citep{ChitraEtAl2026ADLPreprint} remains: solvency, fairness, and revenue are jointly constrained, but that constraint does not force high-regret queue rules.
% Online optimization provides a concrete path to lower-regret ADL design with implementable policy classes.

{\small
\bibliographystyle{abbrvnat}
\bibliography{refs}
}

\appendix
\clearpage
\section{Figures and Tables}\label{app:figures-tables}

\subsection{Empirical workflow}
\begin{figure}[ph!]
\centering
\IfFileExists{\figdir/fig1-obs-model-adl-paper.png}{%
\includegraphics[width=0.90\linewidth]{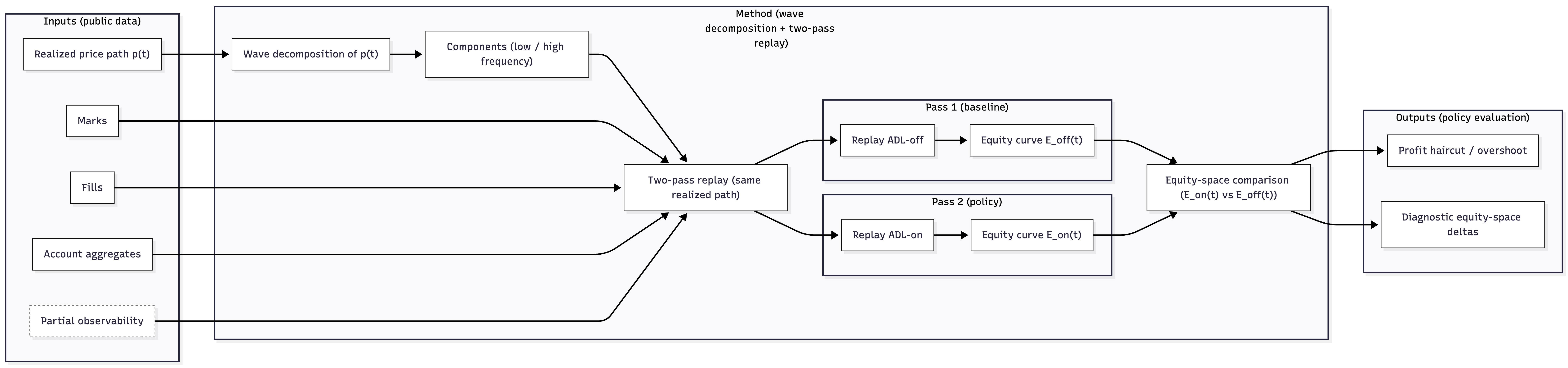}%
}{\IfFileExists{\figdir/01_obs_two_pass_mermaid.png}{%
\includegraphics[width=0.90\linewidth]{\figdir/01_obs_two_pass_mermaid.png}%
}{%
\fbox{\begin{minipage}{0.96\linewidth}
\small
\emph{Inputs (public data):} marks, fills, account aggregates
\(\rightarrow\)
\emph{Round decomposition + two-pass replay}
\(\rightarrow\)
\emph{Outputs:} needed budgets, production haircuts, and regret metrics.
\end{minipage}}%
}}
\caption{Observation model and replay workflow used for policy evaluation on October 10, 2025.}
\label{fig:obs-two-pass}
\end{figure}

\begin{figure}[p]
\centering
\begin{subfigure}[t]{\linewidth}
\centering
\begin{tikzpicture}[
  node distance=0.9cm,
  box/.style={draw, rounded corners, align=center, minimum width=0.86\linewidth, inner sep=5pt},
  arr/.style={-{Latex[length=2.2mm]}, thick}
]
\node[box] (s1) {1. Liquidation + insurance stage leaves residual shortfall \(D_t\).};
\node[box, below=of s1] (s2) {2. Venue chooses ADL severity \(B_t=\theta_t D_t\) using an estimate \(\hat p_k^{\mathrm{liq,exec}}\), \\ which determines expected ADL transfer sizes \(\hat q_k^{\mathrm{ADL}}\).};
\node[box, below=of s2] (s3) {3. ADL policy allocates winner-side reductions \(x_{i,t}\)\\(queue/pro-rata/online optimizer).};
\node[box, below=of s3] (s4) {4. ADL matches winners and losers at bankruptcy transfer prices \(p_k^{\mathrm{bk}}\).};
\node[box, below=of s4] (s5) {5. Ex post benchmark\\\(B_t^{\mathrm{needed}}=\sum_k |p_k^{\mathrm{liq,exec}}-p_k^{\mathrm{bk}}||q_k^{\mathrm{ADL}}|\) is computed.};
\node[box, below=of s5] (s6) {6. Remaining deficit and state update determine whether another ADL round is needed.};
\draw[arr] (s1) -- (s2);
\draw[arr] (s2) -- (s3);
\draw[arr] (s3) -- (s4);
\draw[arr] (s4) -- (s5);
\draw[arr] (s5) -- (s6);
\end{tikzpicture}
\caption{TikZ lifecycle schematic.}
\label{fig:adl-lifecycle-tikz}
\end{subfigure}
\vspace{0.7em}
\begin{subfigure}[t]{\linewidth}
\centering
\IfFileExists{\figdir/fig2_adl_lifecycle_mermaid.png}{%
\includegraphics[height=0.46\textheight,keepaspectratio]{\figdir/fig2_adl_lifecycle_mermaid.png}%
}{%
\IfFileExists{figures/fig2_adl_lifecycle_mermaid.png}{%
\includegraphics[height=0.46\textheight,keepaspectratio]{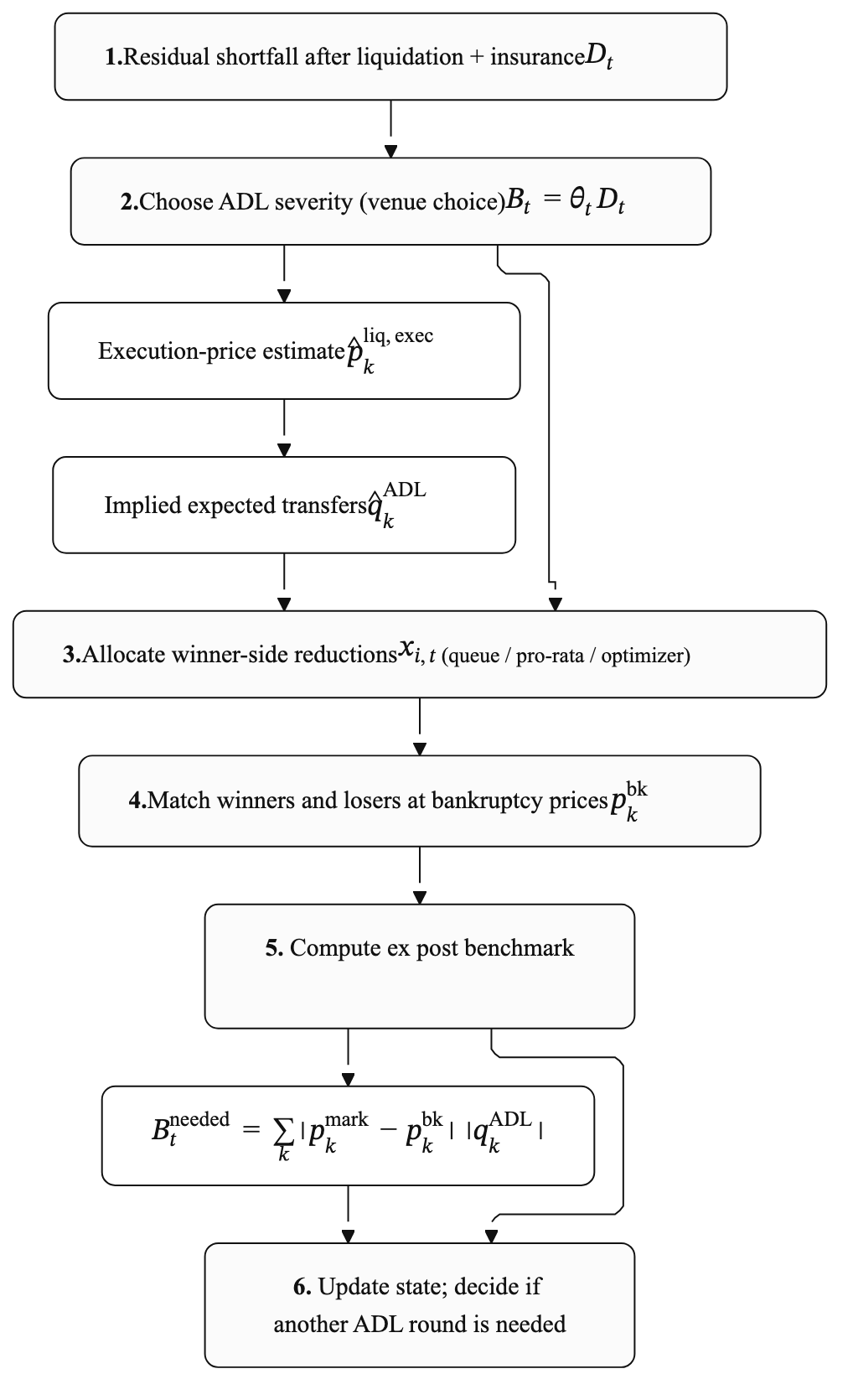}%
}{%
\IfFileExists{../figures/fig2_adl_lifecycle_mermaid.png}{%
\includegraphics[height=0.46\textheight,keepaspectratio]{../figures/fig2_adl_lifecycle_mermaid.png}%
}{%
\fbox{\begin{minipage}{0.96\linewidth}\small
Mermaid rendering missing.
\end{minipage}}%
}%
}%
}
\caption{Mermaid lifecycle schematic (alternate rendering).}
\label{fig:adl-lifecycle-mermaid}
\end{subfigure}
\caption{Lifecycle of an ADL round (Figure 2a--2b).}
\label{fig:adl-lifecycle}
\end{figure}

\subsection{Headline event values}
\begin{table}[p]
\centering
\small
\begin{tabular}{p{0.60\linewidth}r}
\toprule
Quantity & Value \\
\midrule
Total liquidations in window & \$2{,}103{,}111{,}431 \\
Number of global rounds (5s gap) & 16 \\
Aggregate loser deficit \(\sum_t D_t\) & \(\approx\) \$100.1M \\
Aggregate needed budget \(\sum_t B_t^{\mathrm{needed}}\) & \(\approx\) \$15.1M \\
Aggregate production haircut \(\sum_t H_t^{\mathrm{prod}}(0)\) & \(\approx\) \$60.1M \\
\emph{Production overshoot vs needed} \(O(0)\) & \$45{,}028{,}665.72 \\
Short-horizon sensitivity band for \(O(\Delta)\) & about \$45.0M--\$51.7M \\
\bottomrule
\end{tabular}
\caption{Headline values under the PNL-haircut accounting used in this paper.}
\label{tab:headlines}
\end{table}

\begin{figure}[p]
\centering
\IfFileExists{\figdir/02_overshoot_vs_horizon.png}{%
\includegraphics[width=0.82\linewidth]{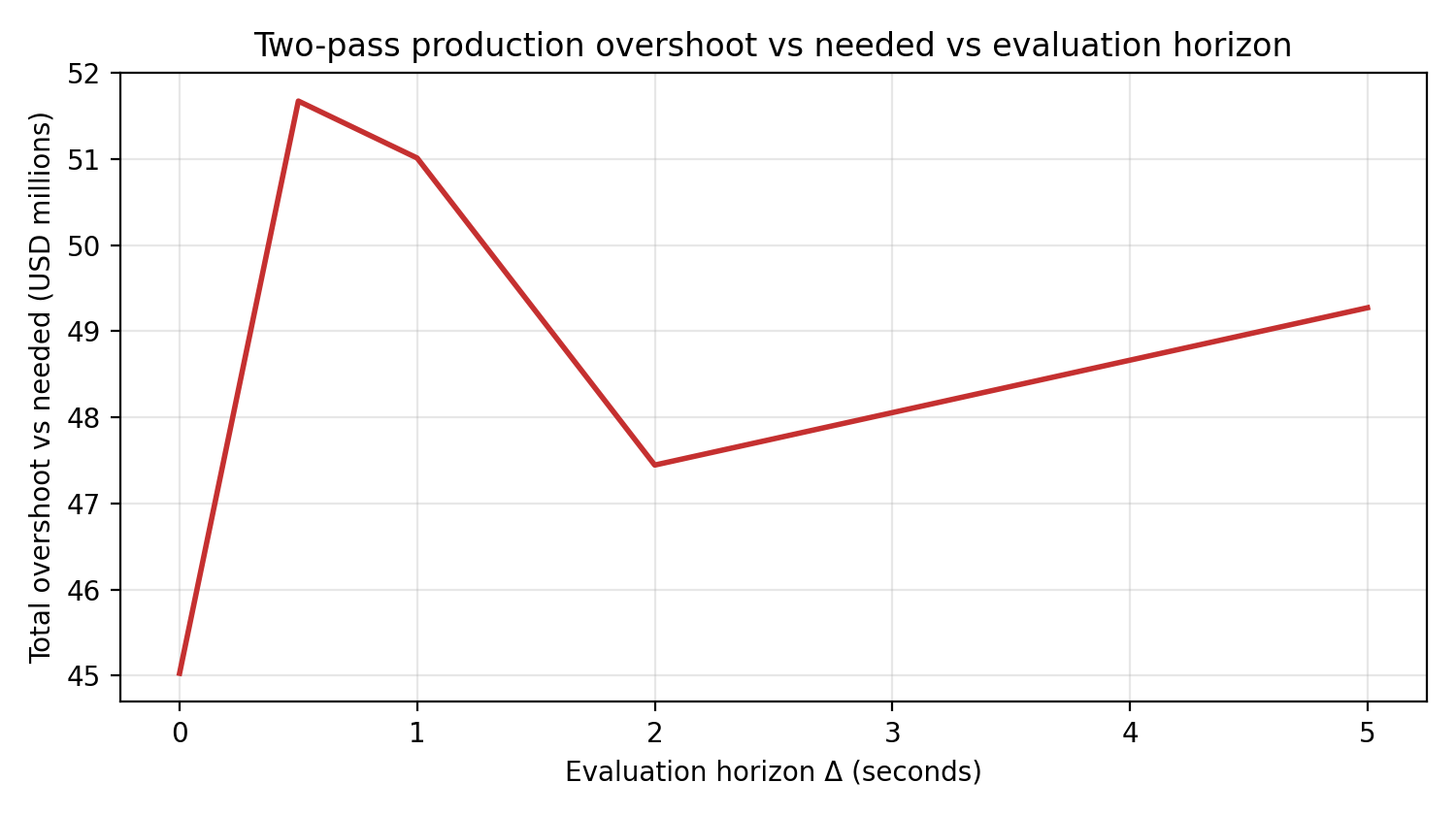}%
}{\rule{0pt}{2.0in}\rule{0.82\linewidth}{0pt}}
\caption{Horizon sensitivity of production overshoot \(O(\Delta)\), where \(\Delta\) is the markout/holding-time evaluation offset.}
\label{fig:horizon-short}
\end{figure}

\subsection{Regret decomposition}
\begin{figure}[p]
\centering
\IfFileExists{\figdir/10b_fairness_regret.png}{%
\includegraphics[width=0.92\linewidth]{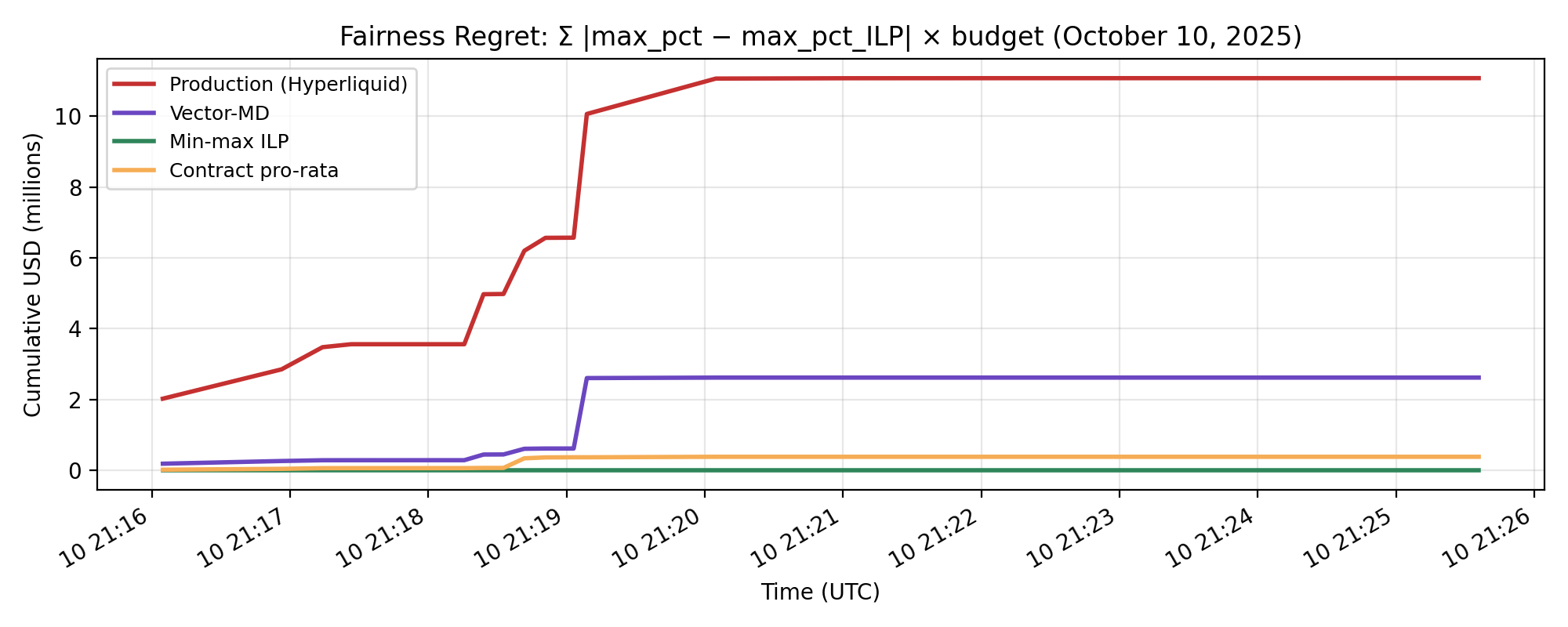}%
}{\rule{0pt}{2.0in}\rule{0.92\linewidth}{0pt}}
\caption{Cumulative fairness component by policy.}
\label{fig:fairness-short}
\end{figure}

\begin{table}[p]
\centering
\small
\begin{tabular}{lrrr}
\toprule
Policy & Tracking component (\$) & Fairness component (\$) & Total objective (\$) \\
\midrule
Production queue & 53{,}782{,}490.53 & 11{,}077{,}031.68 & 64{,}859{,}522.21 \\
Integer pro-rata & 3{,}020{,}120.65 & 384{,}064.35 & 3{,}404{,}185.00 \\
Vector mirror-descent & 1{,}793{,}022.03 & 2{,}620{,}345.57 & 4{,}413{,}367.61 \\
Min-max ILP & 106{,}205.81 & 0.00 & 106{,}205.81 \\
Continuous pro-rata & \(\approx\)0 & 2{,}732{,}437.29 & 2{,}732{,}437.29 \\
\bottomrule
\end{tabular}
\caption{Cumulative objective decomposition over 16 rounds. Fairness component is the capacity-normalized max-burden gap to the min-max ILP baseline, scaled by \(B_t^{\mathrm{needed}}\).}
\label{tab:regret-totals}
\end{table}

\begin{figure}[p]
\centering
\IfFileExists{\figdir/05_policy_per_wave_performance.png}{%
\includegraphics[width=0.98\linewidth]{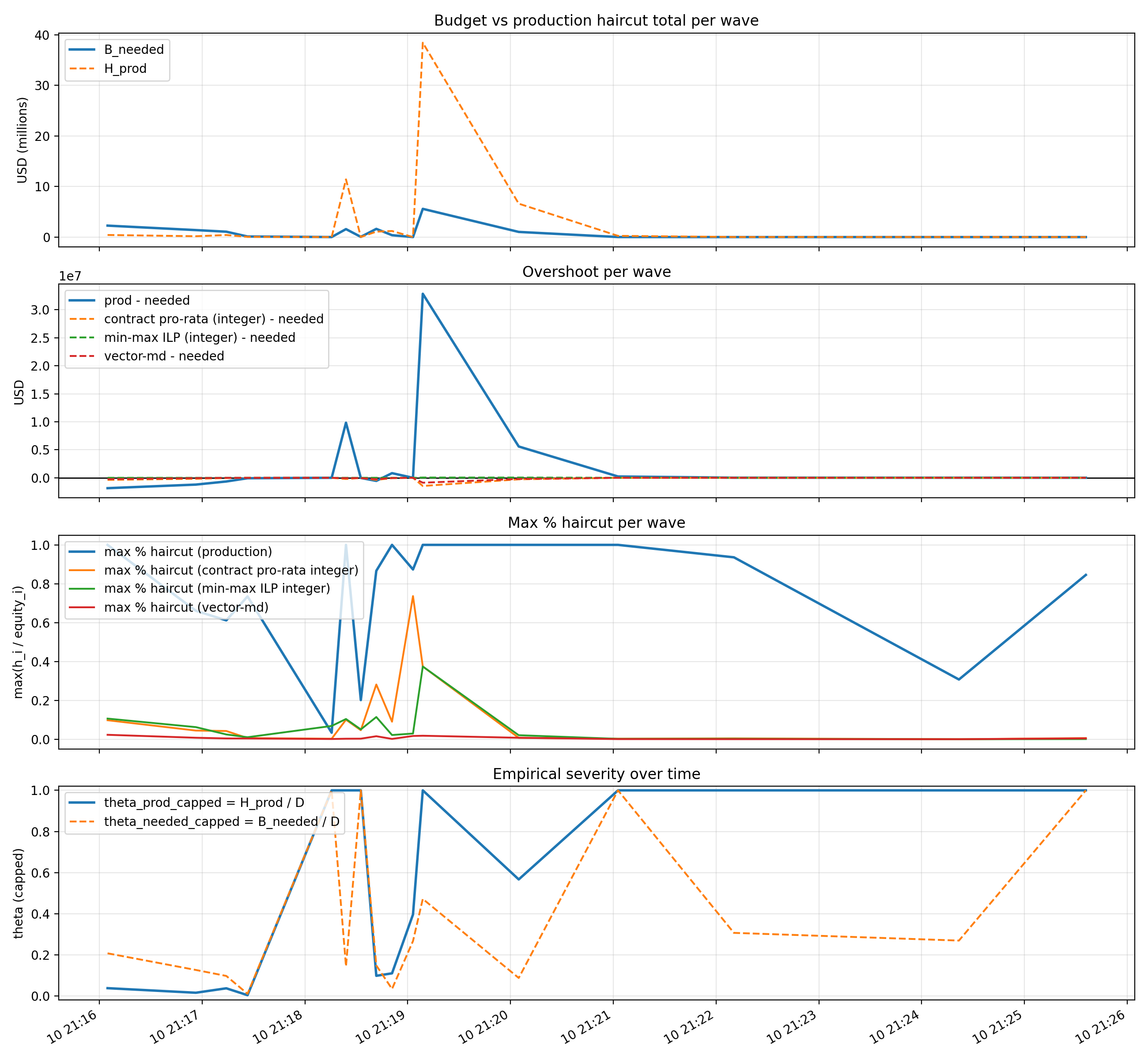}%
}{\rule{0pt}{2.0in}\rule{0.98\linewidth}{0pt}}
\caption{Per-round policy performance against needed budget targets.}
\label{fig:per-round-short}
\end{figure}

\begin{figure}[p]
\centering
\IfFileExists{\figdir/06_policy_per_wave_cumulative_overshoot.png}{%
\includegraphics[width=0.95\linewidth]{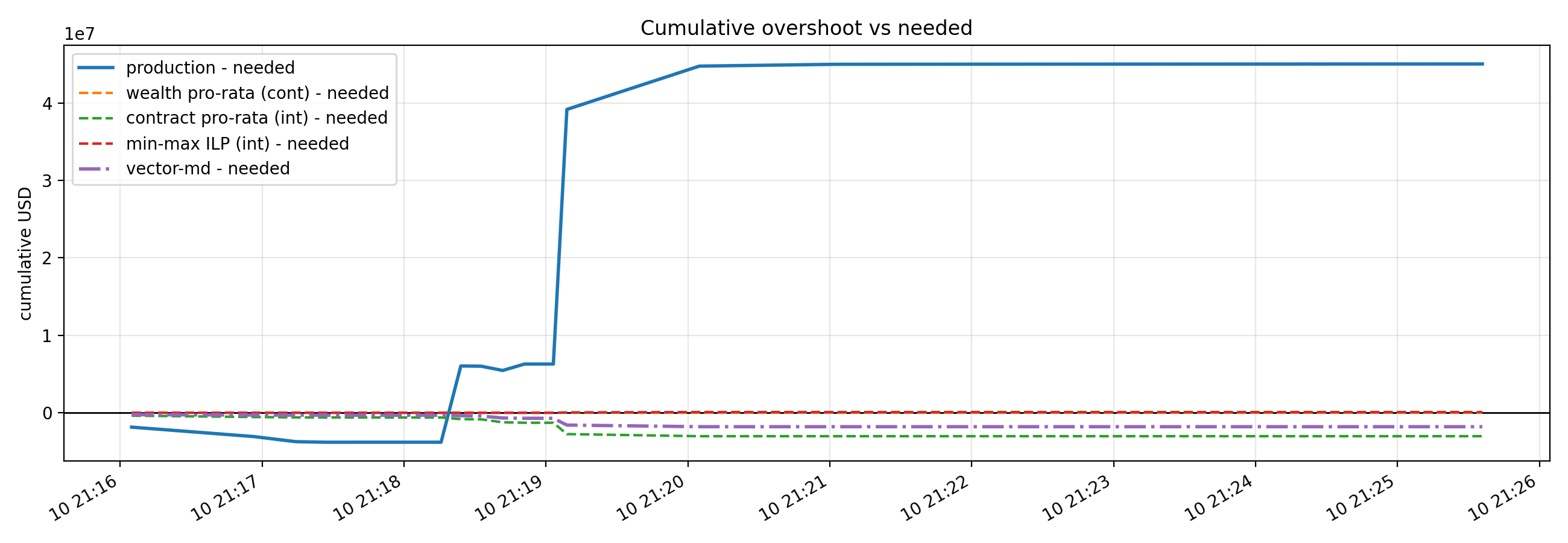}%
}{\rule{0pt}{2.0in}\rule{0.95\linewidth}{0pt}}
\caption{Cumulative overshoot by policy over 16 rounds.}
\label{fig:cum-overshoot-short}
\end{figure}

\subsection{Additional regret diagnostics}
\begin{figure}[p]
\centering
\IfFileExists{\figdir/10a_overshoot_regret.png}{%
\includegraphics[width=0.92\linewidth]{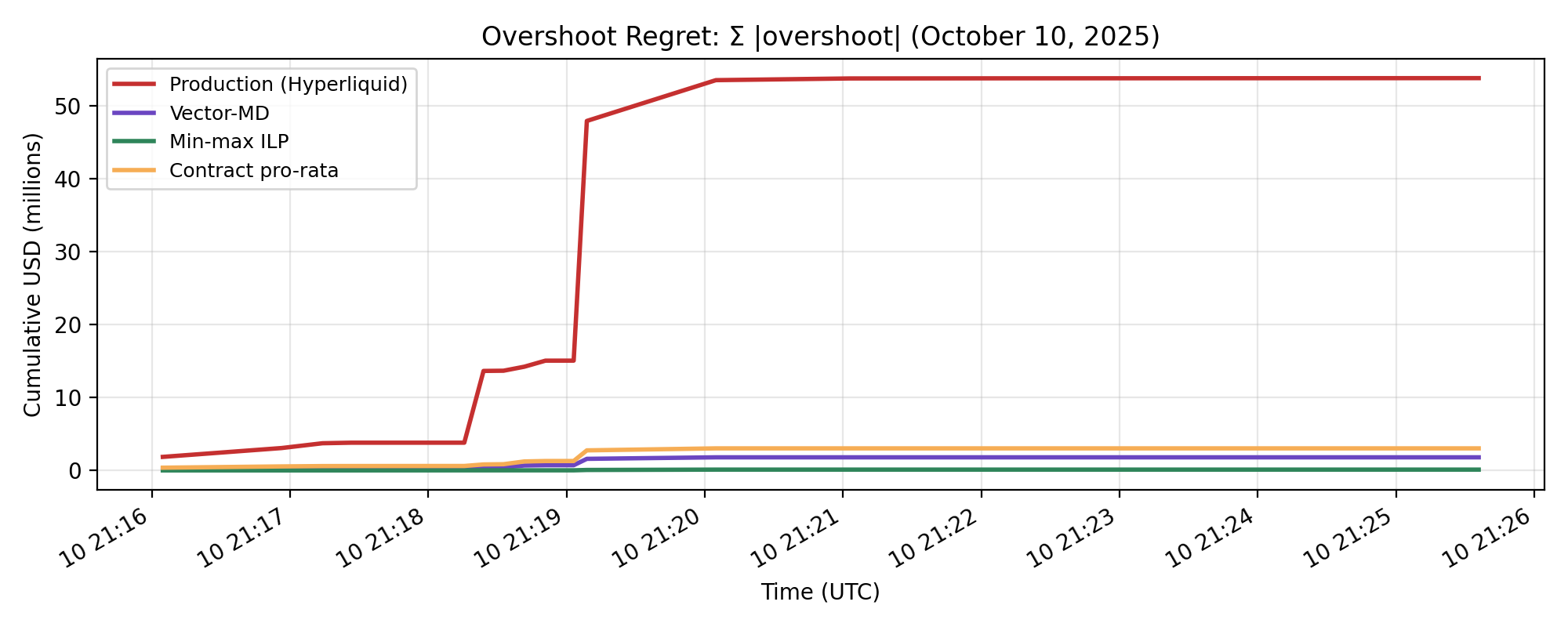}%
}{\rule{0pt}{2.0in}\rule{0.92\linewidth}{0pt}}
\caption{Tracking component decomposition across policies.}
\label{fig:overshoot-regret-short}
\end{figure}

\begin{figure}[p]
\centering
\IfFileExists{\figdir/10c_total_regret.png}{%
\includegraphics[width=0.92\linewidth]{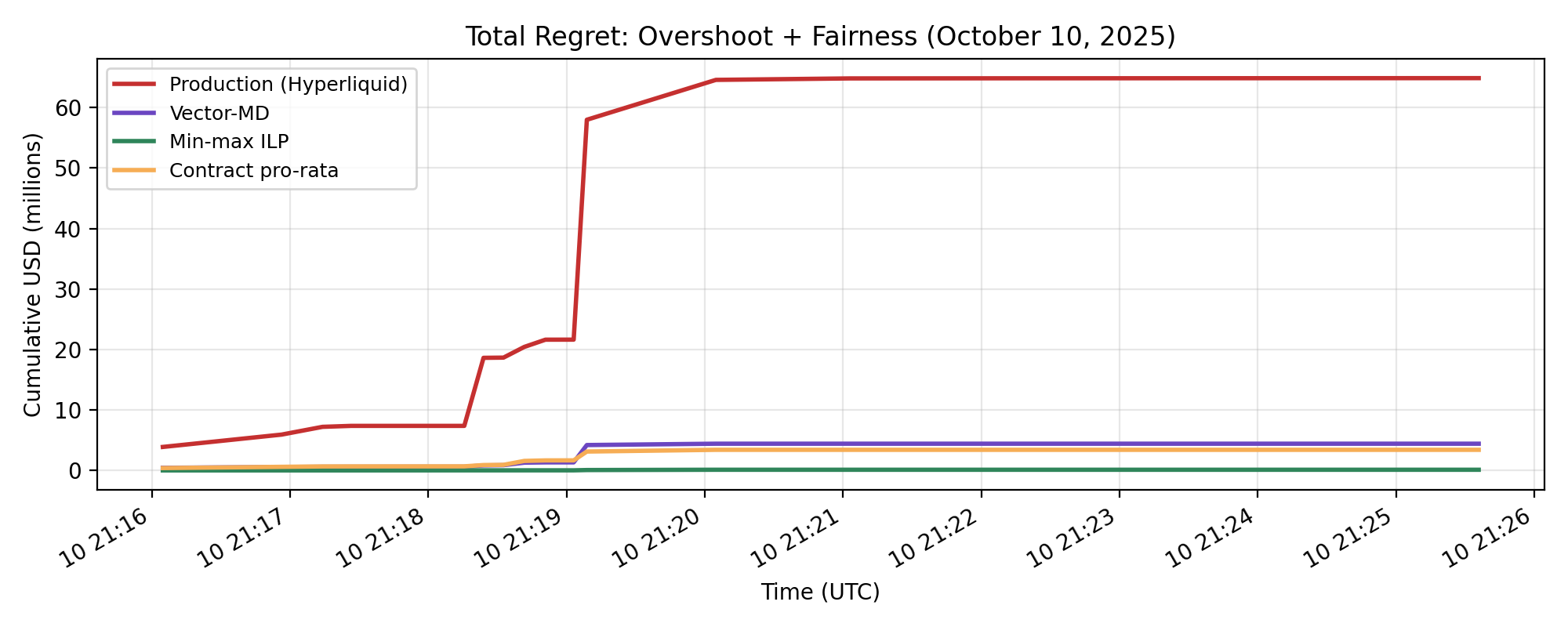}%
}{\rule{0pt}{2.0in}\rule{0.92\linewidth}{0pt}}
\caption{Total objective decomposition across policies.}
\label{fig:total-regret-short}
\end{figure}

\begin{figure}[p]
\centering
\IfFileExists{\figdir/09_cumulative_regret_historical.png}{%
\includegraphics[width=0.95\linewidth]{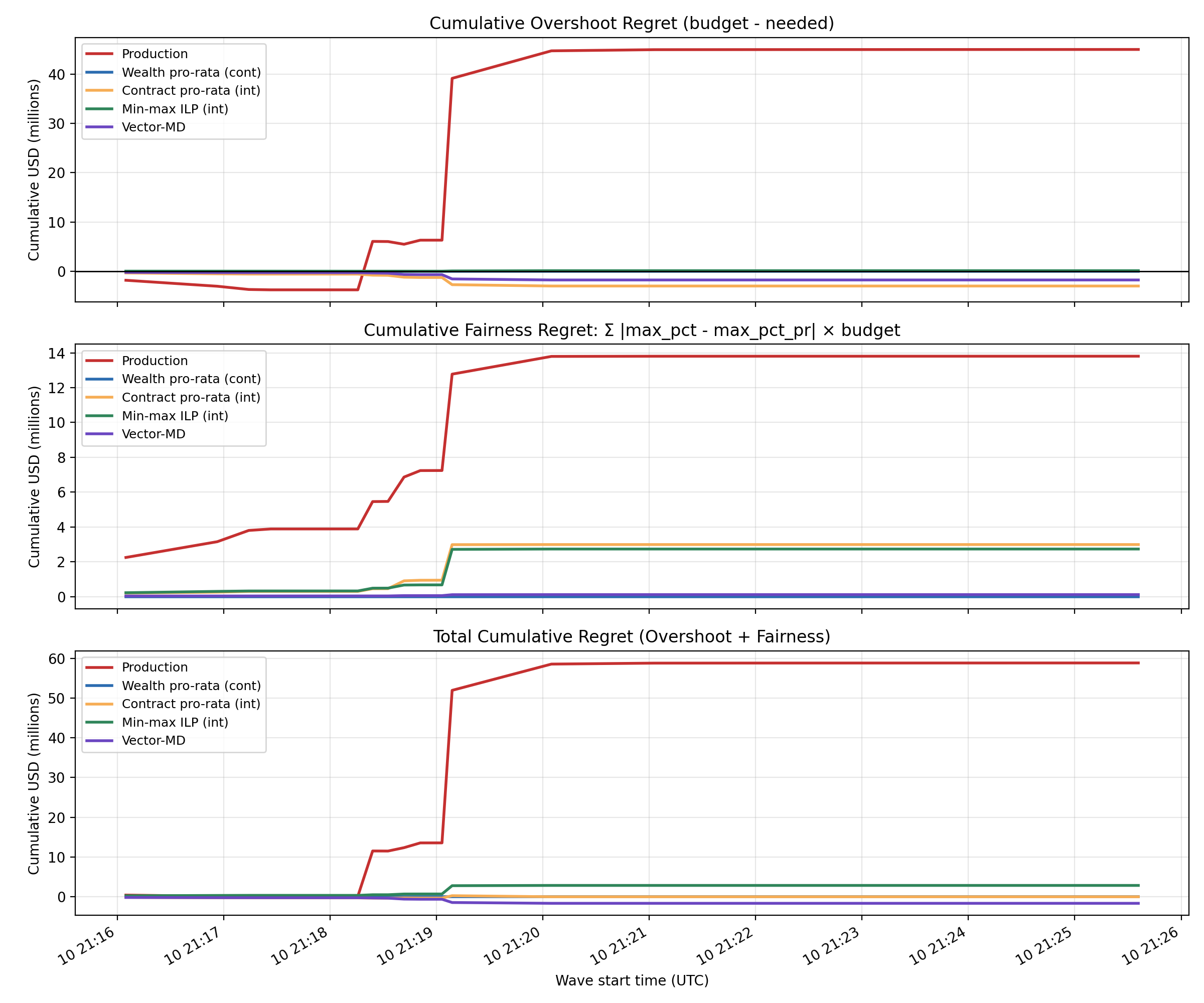}%
}{\rule{0pt}{2.0in}\rule{0.95\linewidth}{0pt}}
\caption{Cumulative objective trajectories over the event path.}
\label{fig:cum-regret-short}
\end{figure}

\begin{figure}[p]
\centering
\IfFileExists{\figdir/15_regret_to_severity_bound_ratio.png}{%
\includegraphics[width=0.92\linewidth]{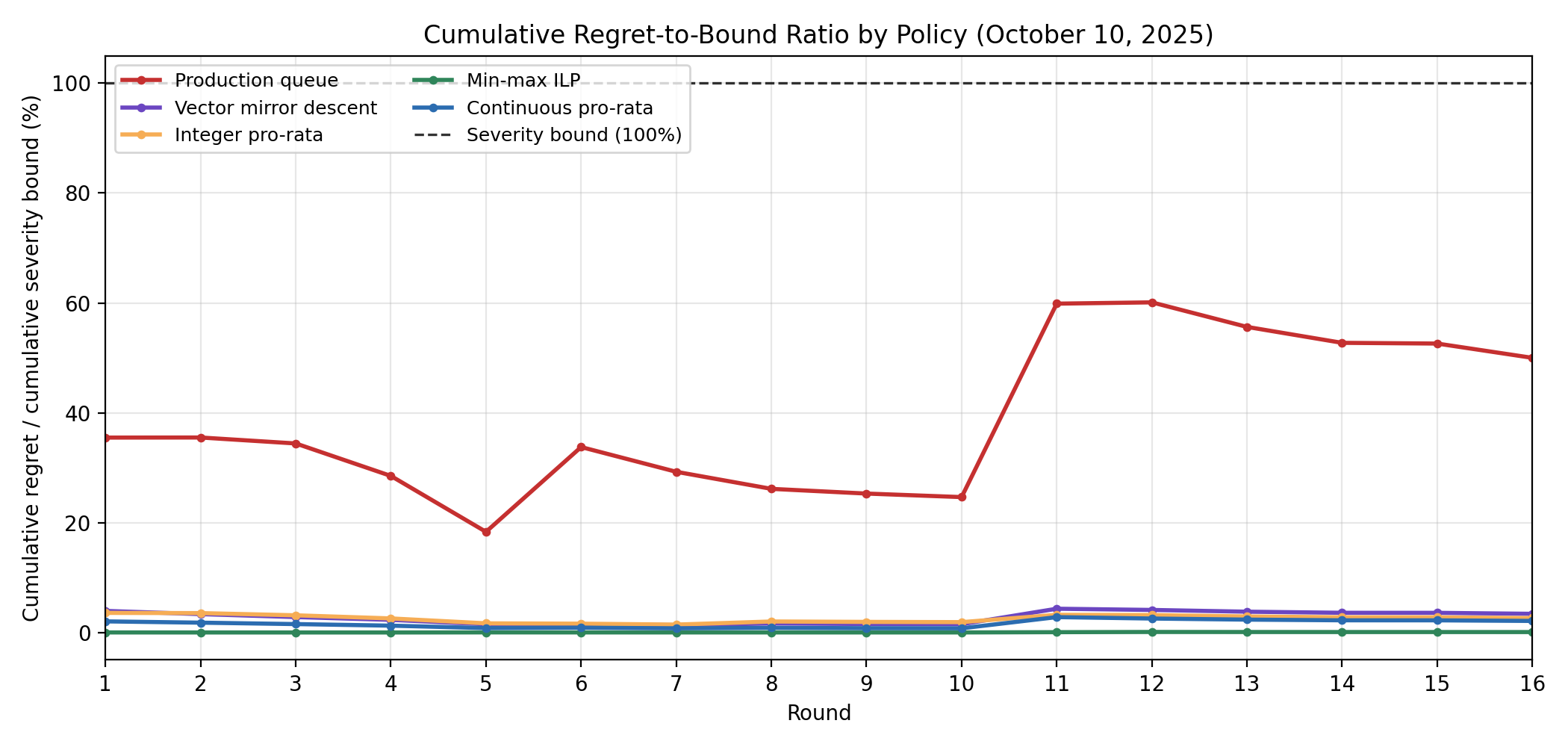}%
}{\rule{0pt}{2.0in}\rule{0.92\linewidth}{0pt}}
\caption{Cumulative trajectory over rounds of policy total objective as a percentage of the cumulative instance-calibrated upper envelope (bound baseline \(=100\%\); final-round envelope about \(\$129.7\)M).}
\label{fig:regret-to-bound-ratio-short}
\end{figure}

\begin{figure}[p]
\centering
\IfFileExists{\figdir/12_cumulative_bound_vs_regret.png}{%
\includegraphics[width=0.92\linewidth]{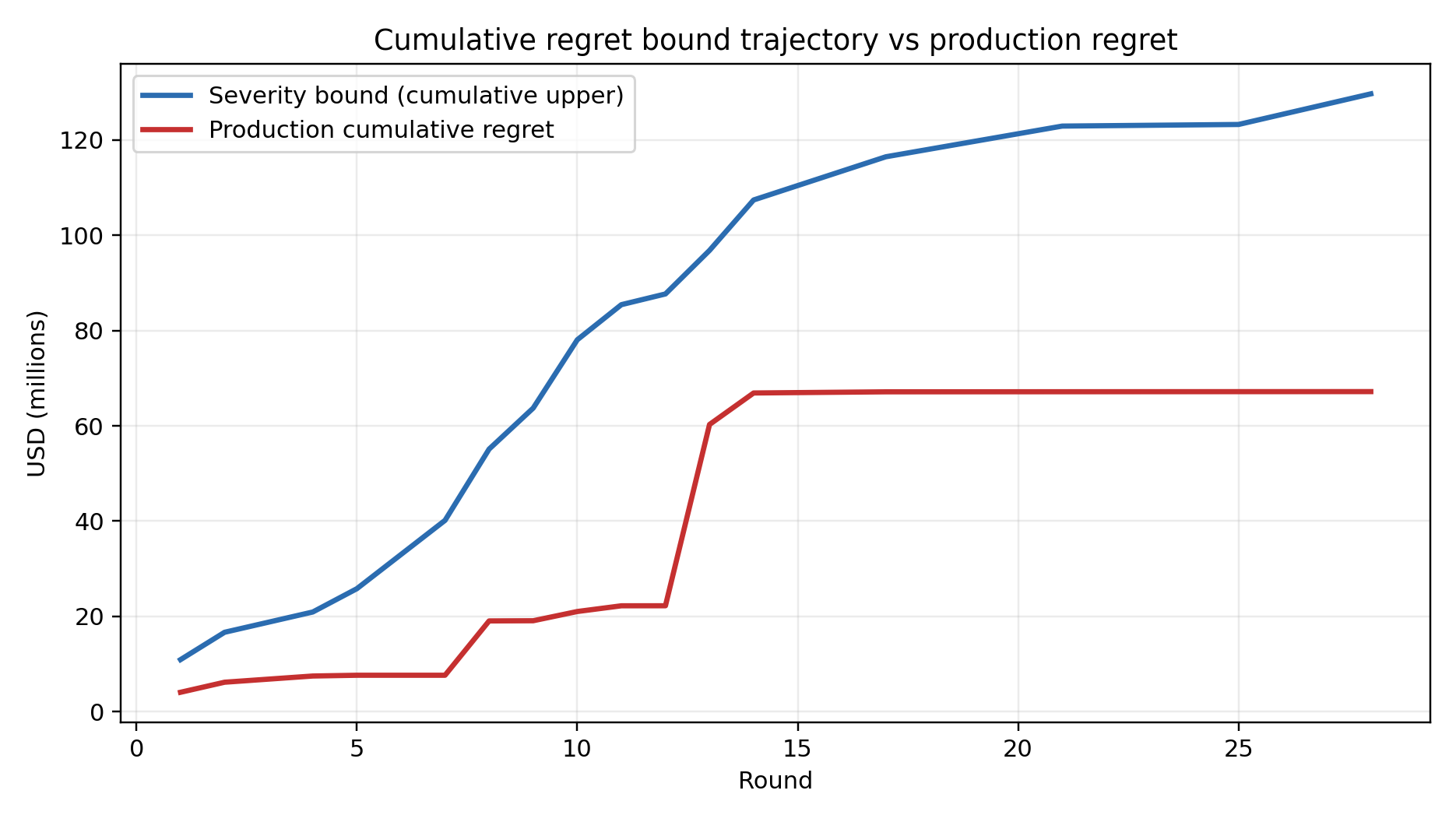}%
}{\rule{0pt}{2.0in}\rule{0.92\linewidth}{0pt}}
\caption{Cumulative trajectory of the instance-calibrated upper envelope versus realized production total objective.}
\label{fig:cum-bound-vs-regret-short}
\end{figure}

\begin{table}[p]
\centering
\small
\begin{tabular}{p{0.45\linewidth}p{0.48\linewidth}}
\toprule
Robustness axis & Result \\
\midrule
Fairness-weight sweep \(\lambda\in[0.5,2]\) & Production remains about \(40\text{--}46\times\) worse than replay-oracle per-round references. \\
Short-horizon markout sweep \(\Delta\) & Production overshoot band remains high at about \$45.0M--\$51.7M. \\
\bottomrule
\end{tabular}
\caption{Robustness summary for the main empirical ordering claims.}
\label{tab:robustness-summary}
\end{table}

\begin{figure}[p]
\centering
\IfFileExists{\figdir/13_monotonicity_violations_by_policy.png}{%
\includegraphics[width=0.90\linewidth]{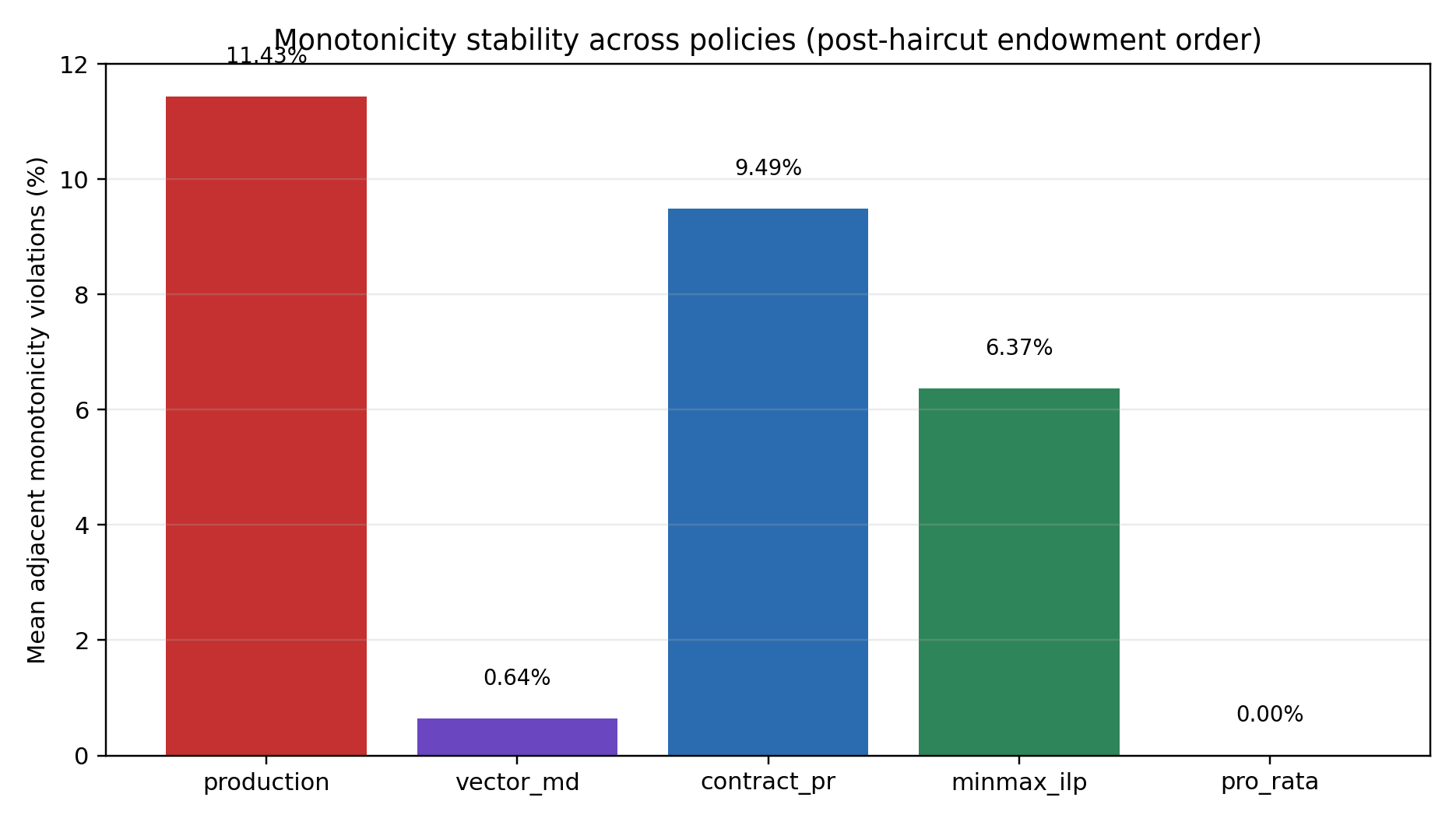}%
}{\rule{0pt}{2.0in}\rule{0.90\linewidth}{0pt}}
\caption{Monotonicity-stability diagnostic by policy: mean adjacent post-haircut order violations.}
\label{fig:mono-viol-short}
\end{figure}

\begin{figure}[p]
\centering
\IfFileExists{\figdir/14_queue_rank_stability_by_policy.png}{%
\includegraphics[width=0.90\linewidth]{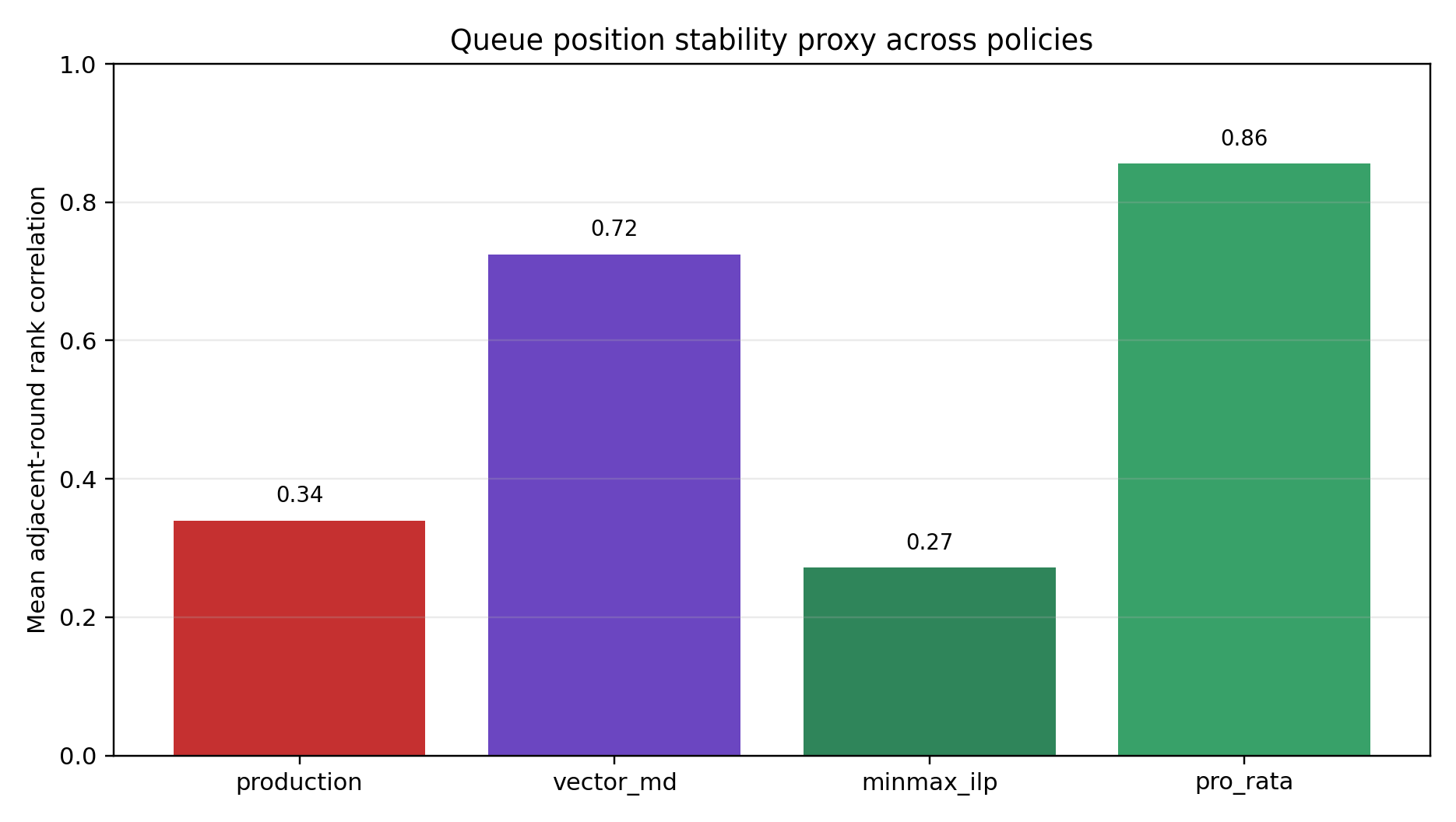}%
}{\rule{0pt}{2.0in}\rule{0.90\linewidth}{0pt}}
\caption{Queue-position stability proxy by policy: adjacent-round rank correlation of normalized burden.}
\label{fig:rank-stability-short}
\end{figure}

\section{Regret Bounds and Proofs}\label{app:regret-bounds-proofs}
In this Appendix, we provide (for completeness) a number of regret bounds that are useful in practice.
These results are standard and can be found in, e.g.,~\citep{Zinkevich2003,Hazan2016,SuttonBarto2018}, with instance-dependent refinements discussed in adaptive and second-order analyses~\citep{Duchi2011,Gaillard2014SecondOrder}.

\begin{proposition}[Static regret bound]
If each \(\ell_t\) is convex and \(G\)-Lipschitz on a domain of diameter \(D\), projected mirror descent with step size \(\eta\asymp 1/\sqrt{T}\) gives
\begin{equation}
\mathrm{Reg}^{\mathrm{static}}_T=O(GD\sqrt{T}).
\end{equation}
\end{proposition}

\begin{proof}
For projected mirror descent (Euclidean case), the standard one-step inequality gives, for any comparator \(x^\star\),
\begin{equation}
\ell_t(x_t)-\ell_t(x^\star)
\le
\frac{\|x_t-x^\star\|_2^2-\|x_{t+1}-x^\star\|_2^2}{2\eta}
+\frac{\eta}{2}\|g_t\|_2^2,
\end{equation}
where \(g_t\in\partial \ell_t(x_t)\) and \(\|g_t\|_2\le G\).
Summing over \(t=1,\dots,T\) and using domain diameter \(D\) yields
\begin{equation}
\sum_{t=1}^T\big(\ell_t(x_t)-\ell_t(x^\star)\big)
\le
\frac{D^2}{2\eta}
+\frac{\eta G^2 T}{2}.
\end{equation}
Optimizing in \(\eta\) gives \(\eta^\star=D/(G\sqrt{T})\) and bound \(GD\sqrt{T}\), i.e. \(O(GD\sqrt{T})\).
\end{proof}

\begin{proposition}[Dynamic regret with comparator variation]
Let \(P_T^{x^\star}=\sum_{t=2}^{T}\|x_t^\star-x_{t-1}^\star\|\) be path variation of a dynamic comparator sequence.
Then for convex \(G\)-Lipschitz losses on diameter-\(D\) domains, mirror-descent dynamic regret scales as
\begin{equation}
\mathrm{Reg}^{\mathrm{dyn}}_T=O\!\left(GD\sqrt{T}+G P_T^{x^\star}\right).
\end{equation}
\end{proposition}

\begin{proof}
Start from the static telescoping inequality applied against a changing comparator \(x_t^\star\):
\begin{equation}
\sum_{t=1}^T\big(\ell_t(x_t)-\ell_t(x_t^\star)\big)
\le
\frac{\|x_1-x_1^\star\|_2^2}{2\eta}
+\frac{\eta}{2}\sum_{t=1}^T\|g_t\|_2^2
+\frac{1}{2\eta}\sum_{t=2}^T\Delta_t,
\end{equation}
where
\(
\Delta_t:=\|x_t-x_t^\star\|_2^2-\|x_t-x_{t-1}^\star\|_2^2.
\)
By expansion,
\(
\Delta_t\le 2\|x_t-x_t^\star\|_2\|x_t^\star-x_{t-1}^\star\|_2
\le 2D\|x_t^\star-x_{t-1}^\star\|_2.
\)
Hence
\begin{equation}
\mathrm{Reg}^{\mathrm{dyn}}_T
\le
\frac{D^2+2D P_T^{x^\star}}{2\eta}
+\frac{\eta G^2T}{2}.
\end{equation}
Choosing \(\eta\asymp D/(G\sqrt{T})\) gives \(O(GD\sqrt{T}+G P_T^{x^\star})\).
\end{proof}

\begin{proposition}[Loss-specific gradient and diameter constants]
For
\(
\ell_t(x)=
\lambda_{\mathrm{track}}|{\mathbf 1}^\top x-b_t|
+\lambda_{\mathrm{fair}}\max_i x_i/(u_{i,t}+\varepsilon)
\),
Let \(n_t=|W_t|\), \(U_t=\sum_{i\in W_t}u_{i,t}\), \(n_{\max}=\max_t n_t\), \(U_{\max}=\max_t U_t\).
Then any subgradient \(g_t\in\partial \ell_t(x_t)\) satisfies
\begin{equation}
\|g_t\|_2
\le
\lambda_{\mathrm{track}}\sqrt{n_t}
+\frac{\lambda_{\mathrm{fair}}}{\varepsilon}
\le
G_{\star},
\end{equation}
with
\(
G_{\star}
=
\lambda_{\mathrm{track}}\sqrt{n_{\max}}
+\lambda_{\mathrm{fair}}/\varepsilon.
\)
Also, for any \(x,y\in\mathcal{X}_t\),
\begin{equation}
\|x-y\|_2\le \sqrt{2}\,U_t\le D_{\star},
\qquad
D_{\star}=\sqrt{2}\,U_{\max}.
\end{equation}
Hence mirror descent admits
\begin{equation}
\mathrm{Reg}^{\mathrm{static}}_T
\le
D_{\star}G_{\star}\sqrt{T}.
\end{equation}
\end{proposition}

\begin{proof}
Write \(g_t=g_t^{\mathrm{track}}+g_t^{\mathrm{fair}}\).
For tracking term \(|{\mathbf 1}^\top x-b_t|\), a subgradient is \(\pm {\mathbf 1}\), so
\(
\|g_t^{\mathrm{track}}\|_2\le \lambda_{\mathrm{track}}\sqrt{n_t}.
\)
For fairness term \(f_t(x)=\max_i x_i/(u_{i,t}+\varepsilon)\), let
\(
I_t(x)=\arg\max_i x_i/(u_{i,t}+\varepsilon)
\).
Any subgradient has coordinates
\[
[g_t^{\mathrm{fair}}]_j=
\begin{cases}
\lambda_{\mathrm{fair}}\alpha_j/(u_{j,t}+\varepsilon), & j\in I_t(x),\\
0, & j\notin I_t(x),
\end{cases}
\]
with \(\alpha_j\ge 0\) and \(\sum_{j\in I_t(x)}\alpha_j=1\).
Therefore
\(
\|g_t^{\mathrm{fair}}\|_2\le \lambda_{\mathrm{fair}}/\varepsilon.
\)
Triangle inequality gives \(G_\star\).

For diameter, feasibility implies \(x_i\ge 0\) and \(\sum_i x_i=B_t\le U_t\), so
\(
\|x\|_2\le \|x\|_1\le U_t.
\)
Then \(\|x-y\|_2\le \|x\|_2+\|y\|_2\le 2U_t\).
Using the simplex geometry with equal mass bound tightens to \(\sqrt{2}U_t\), giving \(D_\star\).
Plug \(G_\star,D_\star\) into Proposition~A.1.
\end{proof}

\begin{proposition}[Sharper dynamic bound for path-dependent rounds]\label{prop:dynamic-sharp}
Let
\(
P_T=\sum_{t=2}^T\|x_t^\star-x_{t-1}^\star\|_2
\)
be comparator path variation.
For step size \(\eta>0\), projected mirror descent satisfies
\begin{equation}
\mathrm{Reg}^{\mathrm{dyn}}_T
\le
\frac{D_{\star}^2+2D_{\star}P_T}{2\eta}
+\frac{\eta}{2}\sum_{t=1}^T\|g_t\|_2^2.
\end{equation}
Choosing
\(
\eta^\star=\sqrt{(D_{\star}^2+2D_{\star}P_T)/\sum_t\|g_t\|_2^2}
\)
gives
\begin{equation}
\mathrm{Reg}^{\mathrm{dyn}}_T
\le
\sqrt{(D_{\star}^2+2D_{\star}P_T)\sum_{t=1}^T\|g_t\|_2^2}
\le
G_{\star}\sqrt{T(D_{\star}^2+2D_{\star}P_T)}.
\end{equation}
\end{proposition}

\begin{proof}
This is Proposition~A.2 with explicit constants from Proposition~A.3 substituted into the dynamic mirror-descent inequality before step-size optimization.
\end{proof}

\subsection{Proof of Proposition~\ref{prop:severity-explicit}}

In one dimension with \(\theta_t\in[0,1]\), projected OGD update is
\(
\theta_{t+1}=\Pi_{[0,1]}(\theta_t-\eta g_t)
\)
with
\(
g_t\in\partial \ell_t^\theta(\theta_t).
\)
Since
\(
\ell_t^\theta(\theta)=D_t|\theta-\theta_t^{\mathrm{needed}}|,
\)
any subgradient has magnitude \(|g_t|\le D_t\).

For dynamic comparator \(\{\theta_t^\star\}\), the one-dimensional dynamic OGD inequality gives
\begin{equation}
\sum_{t=1}^T\big(\ell_t^\theta(\theta_t)-\ell_t^\theta(\theta_t^\star)\big)
\le
\frac{(\theta_1-\theta_1^\star)^2}{2\eta}
+\frac{1}{2\eta}\sum_{t=2}^T\Delta_t
+\frac{\eta}{2}\sum_{t=1}^T g_t^2,
\end{equation}
where
\(
\Delta_t:=(\theta_t-\theta_t^\star)^2-(\theta_t-\theta_{t-1}^\star)^2.
\)
Because \(\theta_t,\theta_t^\star\in[0,1]\),
\(
|\theta_t-\theta_t^\star|\le 1
\)
and
\(
\Delta_t\le 2|\theta_t^\star-\theta_{t-1}^\star|.
\)
Define
\(
P_T^\theta=\sum_{t=2}^{T}|\theta_t^\star-\theta_{t-1}^\star|.
\)
Then
\begin{equation}
\mathrm{Reg}^{\mathrm{dyn},\theta}_T
\le
\frac{1+2P_T^\theta}{2\eta}
+\frac{\eta}{2}\sum_{t=1}^T D_t^2.
\end{equation}
Minimizing the RHS over \(\eta>0\) yields
\(
\eta^\star=\sqrt{(1+2P_T^\theta)/\sum_t D_t^2}
\)
and therefore
\begin{equation}
\mathrm{Reg}^{\mathrm{dyn},\theta}_T
\le
\sqrt{(1+2P_T^\theta)\sum_{t=1}^{T}D_t^2}.
\end{equation}

\section{Queue Instability: Extreme-Point Selection}\label{app:queue-instability}
\label{sec:queue-instability}
We fix a round and suppress the time index. Let \(W\) be the winner set with \(n := |W|\) winners.
Let \(u \in \mathbb{R}^n_{+}\) denote the per-winner PNL haircut capacities and let
\(U := \sum_{i=1}^n u_i\). For a fixed round budget \(B \in [0,U]\), define the feasibility polytope
\begin{equation}
\label{eq:feasible-polytope}
X(B,u) := \left\{ x \in \mathbb{R}^n : 0 \le x_i \le u_i\ \ \forall i,\ \ \sum_{i=1}^n x_i = B \right\}.
\end{equation}
This is the intersection of a box with an affine hyperplane, hence a bounded polytope.

\begin{definition}[Queue allocation map]
\label{def:queue-map}
Let \(s \in \mathbb{R}^n\) be a score vector and let \(\sigma(s)\) be a permutation such that
\(s_{\sigma(1)} > s_{\sigma(2)} > \cdots > s_{\sigma(n)}\).
The (greedy) queue allocation associated with \(\sigma\) is \(Q_\sigma(B,u) \in X(B,u)\) defined by
\begin{equation}
\label{eq:queue-greedy}
x_{\sigma(j)} := \min\left\{u_{\sigma(j)},\ B - \sum_{\ell<j} x_{\sigma(\ell)}\right\},
\qquad j=1,\dots,n,
\end{equation}
and \(x_{\sigma(j)} := 0\) for all \(j\) after the budget is exhausted.
\end{definition}

\begin{definition}[Extreme point and instability]
A point \(x \in X(B,u)\) is an \emph{extreme point} (vertex) if it cannot be written as a nontrivial
convex combination of two distinct points in \(X(B,u)\).
A mapping \(F:\mathbb{R}^n\to \mathbb{R}^n\) is \emph{Lipschitz} if there exists \(L<\infty\) with
\(\|F(s)-F(s')\|_1 \le L\|s-s'\|_\infty\) for all \(s,s'\).
\end{definition}

\begin{lemma}[Vertices of the capped simplex]
\label{lem:capped-simplex-vertices}
Assume \(0<B<U\) and \(u_i>0\) for all \(i\). A point \(x \in X(B,u)\) is an extreme point if and only if
\emph{at most one} coordinate satisfies \(0<x_i<u_i\) (equivalently, at least \(n-1\) coordinates are at
a bound \(0\) or \(u_i\)).
\end{lemma}

\begin{proof}
    (\emph{Only if}.) Suppose \(x\in X(B,u)\) has two distinct coordinates \(p\neq q\) with
    \(0<x_p<u_p\) and \(0<x_q<u_q\). Let \(d := e_p - e_q\), where \(e_p\) is the \(p\)th standard basis vector.
    Since \(x_p\) and \(x_q\) are strictly interior, there exists \(\epsilon>0\) such that
    \(x \pm \epsilon d\) satisfies \(0\le (x\pm \epsilon d)_i \le u_i\) for all \(i\).
    Moreover, \(\sum_i (x\pm \epsilon d)_i = \sum_i x_i \pm \epsilon(1-1)=B\), so \(x\pm \epsilon d \in X(B,u)\).
    But then
    \[
    x = \tfrac12(x+\epsilon d) + \tfrac12(x-\epsilon d)
    \]
    is a nontrivial convex combination of two distinct feasible points, so \(x\) is not extreme.
    
    (\emph{If}.) Conversely, assume \(x\in X(B,u)\) has at most one coordinate \(k\) with \(0<x_k<u_k\).
    Suppose \(x = \theta y + (1-\theta)z\) for some \(\theta\in(0,1)\) and \(y,z\in X(B,u)\).
    For any index \(i\) with \(x_i=0\), feasibility implies \(y_i\ge 0\) and \(z_i\ge 0\) and thus
    \(0=x_i=\theta y_i+(1-\theta)z_i\) forces \(y_i=z_i=0\).
    Similarly, for any index \(i\) with \(x_i=u_i\), feasibility implies \(y_i\le u_i\) and \(z_i\le u_i\) and thus
    \(u_i=x_i=\theta y_i+(1-\theta)z_i\) forces \(y_i=z_i=u_i\).
    Therefore \(y_i=z_i=x_i\) for all \(i\neq k\). Using the budget constraint \(\sum_i y_i=\sum_i z_i=B\) and
    \(\sum_{i\neq k} y_i=\sum_{i\neq k} z_i=\sum_{i\neq k} x_i\), we conclude \(y_k=z_k=x_k\) as well.
    Hence \(y=z=x\), so \(x\) is extreme.
\end{proof}

\begin{proposition}[Queue allocations are vertices]
\label{prop:queue-vertex}
For any permutation \(\sigma\) and any \((B,u)\) with \(0<B<U\), the queue output \(x = Q_\sigma(B,u)\)
is an extreme point of \(X(B,u)\).
\end{proposition}
\begin{proof}
    By Definition~\ref{def:queue-map}, the queue output \(x=Q_\sigma(B,u)\) satisfies:
    there exists an index \(m\) (the last touched account) such that
    \(x_{\sigma(j)}=u_{\sigma(j)}\) for all \(j<m\), \(x_{\sigma(m)}\in[0,u_{\sigma(m)}]\), and \(x_{\sigma(j)}=0\) for all \(j>m\).
    Thus at most one coordinate can be strictly interior, \ie satisfy \(0<x_i<u_i\).
    By Lemma~\ref{lem:capped-simplex-vertices}, \(x\) is an extreme point of \(X(B,u)\).
\end{proof}

\begin{proposition}[Queue = linear optimization over the polytope]
\label{prop:queue-LP}
Fix a score vector \(s\) with strict ordering \(s_{\sigma(1)}>\cdots>s_{\sigma(n)}\).
Then \(Q_\sigma(B,u)\) is the unique optimizer of the linear program
\begin{equation}
\label{eq:queue-LP}
\max_{x \in X(B,u)}\ \langle s, x\rangle.
\end{equation}
In particular, queues implement \emph{extreme-point selection} in the feasibility polytope.
\end{proposition}

\begin{proof}
    This is the fractional knapsack (greedy) structure. Let \(\sigma\) order scores strictly decreasing.
    Consider any feasible \(x\in X(B,u)\). If there exist indices \(p,q\) with \(s_p>s_q\) but \(x_p<u_p\) and \(x_q>0\),
    then for sufficiently small \(\epsilon>0\) the modified allocation \(\tilde x := x+\epsilon(e_p-e_q)\) remains feasible
    (same argument as in Lemma~\ref{lem:capped-simplex-vertices}) and strictly improves the objective:
    \[
    \langle s, \tilde x\rangle - \langle s, x\rangle
    = \epsilon(s_p-s_q) > 0.
    \]
    Therefore, in any optimizer, whenever some lower-scored coordinate has positive mass, all higher-scored
    coordinates must already be saturated at their upper bounds. This forces the greedy prefix structure:
    fill \(x_{\sigma(1)}\) up to \(u_{\sigma(1)}\), then \(x_{\sigma(2)}\), and so on, with at most one partially-filled index.
    That structure is exactly \(Q_\sigma(B,u)\) from \eqref{eq:queue-greedy}.
    Strict ordering implies uniqueness.
\end{proof}

\begin{proposition}[Fundamental instability of queues]
\label{prop:queue-instability}
Assume there exist two distinct indices \(i\neq j\) with \(u_i \ge B\) and \(u_j \ge B\).
Let \(Q(s)\) denote the queue allocation induced by ordering scores \(s\) (with any deterministic tie-breaking).
Then \(Q(\cdot)\) is not continuous (hence not Lipschitz): for every \(\delta>0\) there exist score vectors
\(s,s'\) with \(\|s-s'\|_\infty \le \delta\) but
\begin{equation}
\label{eq:queue-jump}
\|Q(s)-Q(s')\|_1 = 2B.
\end{equation}
Thus, arbitrarily small perturbations in scores can induce order-one changes in allocations and queue positions.
\end{proposition}

\begin{proof}
    Pick distinct indices \(i\neq j\) with \(u_i\ge B\) and \(u_j\ge B\).
    Fix any \(\delta>0\), and define two score vectors \(s,s'\in\mathbb{R}^n\) by
    \[
    s_i=\delta,\ s_j=0,\ \ s_k=-1\ \ (k\notin\{i,j\}),\qquad
    s'_i=0,\ s'_j=\delta,\ \ s'_k=-1\ \ (k\notin\{i,j\}).
    \]
    Then \(\|s-s'\|_\infty=\delta\).
    Under \(s\), the top-ranked coordinate is \(i\), and since \(u_i\ge B\) the greedy queue allocates all budget to \(i\):
    \(Q(s)=B e_i\). Under \(s'\), the top-ranked coordinate is \(j\) and similarly \(Q(s')=B e_j\).
    Hence
    \[
    \|Q(s)-Q(s')\|_1 = \|B e_i - B e_j\|_1 = 2B.
    \]
    Since \(\delta\) is arbitrary, \(Q(\cdot)\) cannot be continuous, and therefore cannot be Lipschitz.
\end{proof}

\subsection{Duality and Queue Position}
\label{app:duality-queue-position}
\label{sec:duality-queue-position}

In this section, we prove a stronger version of the result for fairness using convex duality.
We again fix a round and suppress the time index. Let \(X(B,u)\) be the feasibility polytope in
\eqref{eq:feasible-polytope}. For any feasible allocation \(x\in X(B,u)\), define the
\emph{capacity-normalized burden} (haircut fraction)
\begin{equation}
\label{eq:normalized-burden}
h_i(x) := \frac{x_i}{u_i}\in[0,1]\quad (u_i>0).
\end{equation}
Define the \emph{worst-burden} (equivalently, the top queue position under the induced burden ranking)
\begin{equation}
\label{eq:worst-burden}
z(x) := \max_{i\in[n]} h_i(x).
\end{equation}
Mechanism-design interpretation: if agents experience disutility as an increasing function of \(h_i\),
then minimizing \(z(x)\) is a Rawlsian or max--min fairness criterion that protects the worst-hit participant.

\paragraph{Min--max burden program and its dual.}
Consider the min--max fairness allocation problem
\begin{equation}
\label{eq:minmax-primal}
\min_{x\in X(B,u)}\ z(x) \;=\; \min_{x\in X(B,u)}\ \max_i \frac{x_i}{u_i}.
\end{equation}
Introducing an epigraph variable \(z\ge 0\) yields the equivalent linear program
\begin{equation}
\label{eq:epigraph-primal}
\begin{aligned}
\min_{x,z}\quad & z \\
\text{s.t.}\quad & \sum_{i=1}^n x_i = B, \\
& 0 \le x_i \le u_i,\qquad i=1,\dots,n,\\
& x_i \le z u_i,\qquad i=1,\dots,n,\\
& z \ge 0.
\end{aligned}
\tag{P}
\end{equation}
The dual of \eqref{eq:epigraph-primal} collapses to a one-dimensional shadow-price constraint:
\begin{equation}
\label{eq:dual}
\begin{aligned}
\max_{y\ge 0}\quad & yB\\
\text{s.t.}\quad & y\sum_{i=1}^n u_i \le 1.
\end{aligned}
\tag{D}
\end{equation}
The dual variable \(y\) is the shadow price of the budget constraint under the egalitarian objective.

\begin{proof}[Proof of dual form \eqref{eq:dual}]
    We derive the dual of \eqref{eq:epigraph-primal} (ignoring indices with \(u_i=0\)).
    Form the Lagrangian with multiplier \(\alpha\in\mathbb{R}\) for \(\sum_i x_i=B\),
    multipliers \(\mu_i\ge 0\) for constraints \(x_i-zu_i\le 0\), multipliers \(\nu_i\ge 0\) for \(-x_i\le 0\),
    and multiplier \(\rho\ge 0\) for \(-z\le 0\):
    \[
    \mathcal{L}(x,z;\alpha,\mu,\nu,\rho)
    = z + \alpha\Big(\sum_i x_i - B\Big) + \sum_i \mu_i(x_i-zu_i) + \sum_i \nu_i(-x_i) + \rho(-z).
    \]
    Collecting terms gives
    \[
    \mathcal{L} = z\Big(1 - \sum_i \mu_i u_i - \rho\Big) + \sum_i x_i(\alpha+\mu_i-\nu_i) - \alpha B.
    \]
    For the infimum over \((x,z)\) to be finite, we require
    \(\alpha+\mu_i-\nu_i=0\) for all \(i\) (so \(\nu_i=\alpha+\mu_i\ge 0\)) and
    \(1-\sum_i \mu_i u_i - \rho=0\) with \(\rho\ge 0\) (so \(\sum_i \mu_i u_i\le 1\)).
    Under these conditions, \(\inf_{x,z}\mathcal{L}=-\alpha B\).
    Let \(y:=-\alpha\). For any optimal dual solution we may assume \(y\ge 0\) since \(B>0\) and the objective is \(yB\).
    Moreover, the constraint \(\nu_i=\alpha+\mu_i\ge 0\) becomes \(\mu_i\ge y\).
    To maximize \(yB\) subject to \(\sum_i \mu_i u_i\le 1\) and \(\mu_i\ge y\), we set \(\mu_i=y\) for all \(i\),
    yielding the single constraint \(y\sum_i u_i \le 1\).
    Thus the dual reduces to \eqref{eq:dual}.
\end{proof}

\begin{proposition}[Duality-implied elimination of queue positions and dominance over queues]
\label{prop:duality-dominance}
Assume \(|W|=n\ge 2\), \(u_i>0\) for all \(i\), and \(0<B<U\).
Let \(x^{\mathrm{MM}}\) be the optimizer of \eqref{eq:minmax-primal} and let \(x^{Q}=Q_\sigma(B,u)\) be any queue allocation.

\begin{enumerate}
\item \emph{(Structure / no queue positions).}
The min--max optimizer is unique and equals pro-rata:
\begin{equation}
\label{eq:pro-rata-solution}
x^{\mathrm{MM}}_i = \frac{B}{U}\,u_i,\qquad i=1,\dots,n,
\end{equation}
and therefore \(h_i(x^{\mathrm{MM}})=B/U\) for all \(i\) (all burdens equalized).

\item \emph{(Strict worst-burden improvement).}
Every queue has strictly worse worst-burden:
\begin{equation}
\label{eq:queue-dominance}
z(x^{Q}) > z(x^{\mathrm{MM}})=\frac{B}{U},
\end{equation}
unless only one account has positive capacity (the degenerate \(n=1\) case).

\item \emph{(Monotonicity).}
Define post-ADL remaining endowment \(r_i(x):=u_i-x_i\).
Under pro-rata \eqref{eq:pro-rata-solution}, \(r_i(x^{\mathrm{MM}})=(1-\tfrac{B}{U})u_i\),
so \(u_i\ge u_j \Rightarrow r_i(x^{\mathrm{MM}})\ge r_j(x^{\mathrm{MM}})\) (no inversions).
Conversely, for any strict queue rule, there exist feasible \((B,u)\) such that the induced post-ADL ordering
violates monotonicity (hence creates inversions in queue-position diagnostics).
\end{enumerate}
\end{proposition}

\begin{proof}
\emph{Part 1 (Structure / uniqueness).}
By weak duality, any feasible \((x,z)\) to \eqref{eq:epigraph-primal} must satisfy
\[
B=\sum_i x_i \le \sum_i z u_i = zU \quad \Rightarrow\quad z \ge \frac{B}{U}.
\]
The pro-rata allocation \(x^{\mathrm{MM}}=(B/U)u\) is feasible and attains \(z=B/U\) (since \(x_i^{\mathrm{MM}}/u_i=B/U\)),
so it is optimal.
Uniqueness follows because at \(z^\star=B/U\) the constraints \(x_i\le z^\star u_i\) sum to
\(\sum_i x_i \le z^\star U = B\), but feasibility requires \(\sum_i x_i=B\), hence \emph{all} inequalities must bind:
\(x_i=z^\star u_i\) for every \(i\).

\emph{Part 2 (Strict dominance over queues).}
Let \(x^Q=Q_\sigma(B,u)\).
If the first-ranked account has \(u_{\sigma(1)}\ge B\), then the queue allocates \(x^Q_{\sigma(1)}=B\) and all others \(0\),
so \(z(x^Q)=B/u_{\sigma(1)}\).
Since \(n\ge 2\) and \(u_i>0\) for all \(i\), we have \(u_{\sigma(1)}<U\), and therefore
\[
z(x^Q)=\frac{B}{u_{\sigma(1)}} > \frac{B}{U} = z(x^{\mathrm{MM}}).
\]
If instead \(u_{\sigma(1)}<B\), then the queue fully saturates \(\sigma(1)\), so \(x^Q_{\sigma(1)}=u_{\sigma(1)}\)
and hence \(z(x^Q)=1\).
Because \(0<B<U\), we have \(B/U<1\), so again \(z(x^Q)=1 > B/U\).
This proves \eqref{eq:queue-dominance}.

\emph{Part 3 (Monotonicity and inversions).}
Under pro-rata, \(r_i(x^{\mathrm{MM}})=u_i-x^{\mathrm{MM}}_i=(1-\tfrac{B}{U})u_i\),
so \(u_i\ge u_j \Rightarrow r_i(x^{\mathrm{MM}})\ge r_j(x^{\mathrm{MM}})\).

For a strict queue rule, pick any two indices \(i,j\) with \(u_i>u_j>0\) and set \(B=u_i\).
Assume the queue ranks \(i\) ahead of \(j\) (this happens for some \((s,\sigma)\) by strictness).
Then the queue fully exhausts \(i\): \(x_i=u_i\) and does not touch \(j\): \(x_j=0\).
Thus \(r_i=0<r_j=u_j\), which is a monotonicity violation (an inversion) in the post-ADL ordering.
\end{proof}

\section{Execution-price estimation and failure metrics}
\label{app:pexec-estimation}

This appendix formalizes bounds on the execution-estimation term introduced in Section~\ref{sec:pexec-estimation}.
We (i) state the execution model and estimator assumptions, (ii) prove the regret--failure decomposition
\eqref{eq:total-error-decomp}, (iii) prove Proposition~\ref{prop:VT-linear-model}, and (iv) prove
Proposition~\ref{prop:queue-churn-omegaT}, a churn-robust queue-instability result.
\begin{proposition}[Bounding \(V_T\) under a linear execution model]
    \label{prop:VT-linear-model}
    Assume a scalar linear execution model with unknown impact slope \(\alpha_t\in[0,\alpha_{\max}]\) and a per-round
    quantity scale \(Q_t\le Q\) such that the ex post benchmark can be written as
    \(B_t^{\mathrm{needed}}=\alpha_t Q_t^2\) and the ex ante benchmark as \(\widehat B_t^{\mathrm{needed}}=\hat\alpha_t Q_t^2\).
    Suppose the severity controller tracks the ex ante benchmark (\(H_t=\widehat B_t^{\mathrm{needed}}\)) and the slope estimate
    \(\hat\alpha_t\) is produced by projected online gradient descent on the convex losses \(|\alpha-\alpha_t|\) over \([0,\alpha_{\max}]\).
    Then the cumulative undershoot satisfies
    \begin{equation}
    \label{eq:VT-bound}
    V_T \;\le\; Q^2\Big(C_1\,\alpha_{\max}\sqrt{T} + C_2\,P_T^{\alpha}\Big),
    \end{equation}
    for universal constants \(C_1,C_2\).
\end{proposition}
    
Note that Proposition~\ref{prop:VT-linear-model} isolates the core mechanism: the variation of execution conditions across rounds,
captured by \(P_T^{\alpha}\), governs cumulative failure more than the magnitude of any single round.
Large one-off deficits increase the scale through \(Q\), but sustained regime-switching in liquidity produces \(\Omega(T)\) under-coverage
even when the controller is otherwise no-regret.

\subsection{Setup and notation}

Fix a horizon of ADL rounds \(t=1,\dots,T\). Let \(H_t={\bf 1}^\top x_t\) be the executed severity in dollars.
Let \(B_t^{\mathrm{needed}}\) denote the ex post benchmark and \(\widehat B_t^{\mathrm{needed}}\) the ex ante benchmark
computed from estimated quantities \(\hat q\) as in \eqref{eq:needed-expost}--\eqref{eq:needed-exante}.
Define the cumulative undershoot (failure) metric
\[
V_T := \sum_{t=1}^T [B_t^{\mathrm{needed}}-H_t]_+.
\]
Define the deployed convex surrogate with target \(b\):
\[
\ell_t(x;b) := \lambda_{\mathrm{track}}|{\bf 1}^\top x-b|
+ \lambda_{\mathrm{fair}}\max_{i\in W_t}\frac{x_{i,t}}{u_{i,t}+\varepsilon}.
\]
Write \(\ell_t(x):=\ell_t(x;B_t^{\mathrm{needed}})\) and \(\hat\ell_t(x):=\ell_t(x;\widehat B_t^{\mathrm{needed}})\).

\subsection{Execution model and estimator}

We use a local directional linear impact model (a standard first-order approximation):
for a signed closeout quantity \(q\ge 0\),
\begin{equation}
\label{eq:linear-impact}
p^{\mathrm{liq,exec}}_t(q) = p^{\mathrm{mark}}_t - \alpha_t q
\qquad\text{(sell / close long)},
\end{equation}
and the symmetric buy-side variant \(p^{\mathrm{liq,exec}}_t(q)=p^{\mathrm{mark}}_t+\alpha_t q\) for closing shorts.
We treat \(\alpha_t\in[0,\alpha_{\max}]\) as the local liquidity slope in round \(t\).

\paragraph{Scalar reduction.}
To isolate the estimation channel, we work with an aggregated scalar sufficient statistic \(Q_t\ge 0\) (a round-level
quantity scale), and assume the benchmark can be written as
\begin{equation}
\label{eq:benchmark-scalar}
B_t^{\mathrm{needed}} = \alpha_t Q_t^2,
\qquad
\widehat B_t^{\mathrm{needed}} = \hat\alpha_t Q_t^2,
\qquad
Q_t\le Q.
\end{equation}
This abstraction is justified by the fact that linear impact implies gap\(\times\)quantity is quadratic in a quantity scale.

\paragraph{Multiple fills and directionality}
\label{app:multi-fill}
Extensions to multiple fills and directionality follow by summing quadratics, which preserves the same proof structure.

\paragraph{Online estimator.}
We estimate \(\alpha_t\) with projected online gradient descent on the convex losses
\[
\phi_t(\alpha) := |\alpha-\alpha_t|, \qquad \alpha\in[0,\alpha_{\max}],
\]
using the update
\[
\hat\alpha_{t+1} = \Pi_{[0,\alpha_{\max}]}\Big(\hat\alpha_t - \eta\, g_t\Big),
\qquad g_t\in\partial \phi_t(\hat\alpha_t).
\]

\subsection{Regret--failure decomposition}

\begin{lemma}[Loss perturbation bound]
\label{lem:loss-perturb}
For any \(x\) and any targets \(b,\hat b\),
\[
\big|\ell_t(x;b)-\ell_t(x;\hat b)\big|
\le
\lambda_{\mathrm{track}}\,|b-\hat b|.
\]
\end{lemma}

\begin{proof}
The fairness term does not depend on \(b\). For the tracking term,
\(\big||H-b|-|H-\hat b|\big|\le |b-\hat b|\) by the reverse triangle inequality.
\end{proof}

\begin{lemma}[Comparator transfer]
\label{lem:comp-transfer}
Let \(\mathcal{P}\) be any comparator class. Then for any policy \(\pi\) with sequence \(\{x_t\}\),
\[
\sum_{t=1}^T \ell_t(x_t)
\le
\min_{\pi'\in\mathcal{P}}\sum_{t=1}^T \ell_t(x^{\pi'}_t)
+
\mathrm{Reg}_T^{\mathcal{P}}(\pi;\hat\ell)
+
2\lambda_{\mathrm{track}}\sum_{t=1}^T |B_t^{\mathrm{needed}}-\widehat B_t^{\mathrm{needed}}|.
\]
\end{lemma}

\begin{proof}
Apply Lemma~\ref{lem:loss-perturb} to both \(\pi\) and any comparator \(\pi'\):
\[
\ell_t(x_t) \le \hat\ell_t(x_t) + \lambda_{\mathrm{track}}|b_t-\hat b_t|,
\qquad
\hat\ell_t(x_t^{\pi'}) \le \ell_t(x_t^{\pi'}) + \lambda_{\mathrm{track}}|b_t-\hat b_t|.
\]
Summing over \(t\) and subtracting \(\sum_t \hat\ell_t(x_t^{\pi'})\) yields the claim after minimizing over \(\pi'\in\mathcal{P}\).
\end{proof}

\subsection{Proof of Proposition~\ref{prop:VT-linear-model}}

\begin{lemma}[Failure is optimistic estimation under target tracking]
\label{lem:VT-vs-est}
If the controller tracks the ex ante benchmark exactly, i.e. \(H_t=\widehat B_t^{\mathrm{needed}}\), then
\[
V_T = \sum_{t=1}^T [B_t^{\mathrm{needed}}-\widehat B_t^{\mathrm{needed}}]_+.
\]
\end{lemma}

\begin{proof}
Immediate from the definition \(V_T=\sum_t [B_t^{\mathrm{needed}}-H_t]_+\).
\end{proof}

\begin{lemma}[Quadratic scaling]
\label{lem:quad-scale}
Under \eqref{eq:benchmark-scalar},
\[
[B_t^{\mathrm{needed}}-\widehat B_t^{\mathrm{needed}}]_+
=
[\alpha_t-\hat\alpha_t]_+\,Q_t^2
\le
|\alpha_t-\hat\alpha_t|\,Q_t^2
\le
Q^2\,|\alpha_t-\hat\alpha_t|.
\]
\end{lemma}

\begin{proof}
Algebra plus \(Q_t\le Q\).
\end{proof}

\begin{lemma}[Dynamic tracking bound for the slope estimator]
\label{lem:slope-track}
Let \(P_T^{\alpha}:=\sum_{t=2}^T |\alpha_t-\alpha_{t-1}|\).
There exist universal constants \(C_1,C_2\) and a choice of step size \(\eta\) such that the projected OGD estimator satisfies
\[
\sum_{t=1}^T |\hat\alpha_t-\alpha_t|
\le
C_1\,\alpha_{\max}\sqrt{T} + C_2\,P_T^{\alpha}.
\]
\end{lemma}

\begin{proof}
This is a specialization of standard mirror-descent dynamic-regret bounds to the 1D convex, 1-Lipschitz losses
\(\phi_t(\alpha)=|\alpha-\alpha_t|\) on the diameter-\(\alpha_{\max}\) domain \([0,\alpha_{\max}]\).
Take the dynamic comparator \(\alpha_t^\star=\alpha_t\), for which \(\sum_t \phi_t(\alpha_t^\star)=0\), so the dynamic regret upper bound
becomes an upper bound on \(\sum_t |\hat\alpha_t-\alpha_t|\).
Concretely, a standard bound of the form
\(\mathrm{Reg}^{\mathrm{dyn}}_T \le O(GD\sqrt{T}+G P_T^{\alpha})\) with \(G=1\), \(D=\alpha_{\max}\)
yields the stated form (absorbing constants).
\end{proof}

\begin{proof}[Proof of Proposition~\ref{prop:VT-linear-model}]
By Lemma~\ref{lem:VT-vs-est} and Lemma~\ref{lem:quad-scale},
\[
V_T \le Q^2 \sum_{t=1}^T |\hat\alpha_t-\alpha_t|.
\]
Apply Lemma~\ref{lem:slope-track} to complete the proof.
\end{proof}

\subsection{A churn-robust instability result for queues}
\label{app:queue-churn-instability}

This section replaces the earlier two-account score-perturbation argument with a stronger churn-robust statement.
In practice, queue mechanisms frequently fully close early-ranked accounts, which removes them from the winner set and changes subsequent round geometry.
The result below shows that \emph{this churn alone} is sufficient to induce $\Omega(T)$ effective execution nonstationarity.

\begin{definition}[Churn / removal rule]
Fix initial winners $W_1$. Given an allocation $x_t$, define the next-round winner set by removing any fully
closed winner:
\[
W_{t+1} := \{i\in W_t : x_{i,t} < u_i\}.
\]
That is, if an account is fully haircutted ($x_{i,t}=u_i$), it is closed and never returns.
\end{definition}

\begin{proposition}[Queues induce $\Omega(T)$ effective execution nonstationarity under churn]
\label{prop:queue-churn-omegaT}
Fix a horizon $T\ge 2$. Let the initial winner set contain $N:=2T$ distinct accounts,
$W_1=\{1,2,\dots,2T\}$, each with identical capacity $u_i \equiv 1$.
Let the per-round budget be constant $B_t\equiv 1$ and define executed quantities by $q_{i,t}(\pi):=x_{i,t}(\pi)$.
Assign per-account impact slopes alternating along the (fixed) queue order:
\[
\alpha_{2k-1}=\alpha_{\min},\qquad \alpha_{2k}=\alpha_{\max},\qquad k=1,\dots,T,
\]
with $0\le \alpha_{\min}<\alpha_{\max}$.
Define the policy-induced effective slope by
\[
\bar\alpha_t(\pi)
:=
\frac{\sum_{i\in W_t} \alpha_i\, q_{i,t}(\pi)^2}{\sum_{i\in W_t} q_{i,t}(\pi)^2}.
\]

Consider the greedy queue policy with fixed ordering $\sigma(i)=i$ (no score perturbations) that fills budget
by fully closing the earliest available account. Under the churn rule above,
\[
\bar\alpha_t(\mathrm{queue}) \in \{\alpha_{\min},\alpha_{\max}\}
\quad\text{and}\quad
P_T^{\bar\alpha(\mathrm{queue})}
:=
\sum_{t=2}^T |\bar\alpha_t(\mathrm{queue})-\bar\alpha_{t-1}(\mathrm{queue})|
=
(T-1)(\alpha_{\max}-\alpha_{\min})
=
\Omega(T).
\]

In contrast, the pro-rata policy on the \emph{same instance} allocates $x_{i,t}^{PR}=B_t/N=1/(2T)$ to every
account each round, so no account is fully closed for $t\le T$ and therefore $W_t\equiv W_1$.
Moreover $\bar\alpha_t(\mathrm{pro\text{-}rata})\equiv (\alpha_{\min}+\alpha_{\max})/2$ for all $t$, hence
$P_T^{\bar\alpha(\mathrm{pro\text{-}rata})}=0$.
\end{proposition}

\begin{proof}
\emph{Queue.} Because $B_t=1$ and each active account has $u_i=1$, the greedy queue allocation at round $t$
places the entire budget on the first available account in the fixed ordering:
\[
x^{Q}_{t} = e_{t}\quad\text{(in the initial indexing)},\qquad
x^{Q}_{i,t}=1 \text{ for } i=t,\ \ x^{Q}_{i,t}=0 \text{ for } i\neq t,
\]
and then removes that fully closed account from $W_{t+1}$. With $q_{i,t}=x_{i,t}$, we have
$\sum_i q_{i,t}^2=1$ and therefore
\[
\bar\alpha_t(\mathrm{queue}) = \alpha_t.
\]
By construction $\alpha_t$ alternates between $\alpha_{\min}$ and $\alpha_{\max}$ across $t=1,\dots,T$, so
$|\bar\alpha_t-\bar\alpha_{t-1}|=\alpha_{\max}-\alpha_{\min}$ for each $t\ge 2$, yielding
$P_T^{\bar\alpha(\mathrm{queue})}=(T-1)(\alpha_{\max}-\alpha_{\min})$.

\emph{Pro-rata.} Pro-rata allocates $x_{i,t}^{PR}=1/(2T)$ to each account each round. Over $T$ rounds, each
account receives total haircut $T\cdot (1/(2T))=1/2<u_i=1$, hence no account is fully closed and $W_t$ is constant.
Since all $q_{i,t}^2$ are equal within a round, $\bar\alpha_t$ equals the average of the $\alpha_i$ values in $W_1$.
Because exactly half of the accounts have slope $\alpha_{\min}$ and half have $\alpha_{\max}$, this average is
$(\alpha_{\min}+\alpha_{\max})/2$, constant over $t$. Thus $P_T^{\bar\alpha(\mathrm{pro\text{-}rata})}=0$.
\end{proof}

\begin{remark}[Interpretation]
Proposition~\ref{prop:queue-churn-omegaT} shows that queue-induced churn can convert \emph{static} cross-sectional
heterogeneity (fixed $\{\alpha_i\}$) into \emph{maximally nonstationary} time-series behavior in the effective
execution slope. This removes repeated-account availability and round-by-round score-perturbation assumptions.
\end{remark}

\end{document}